\documentclass[12pt]{article}
\usepackage{amsmath}
\usepackage{graphicx}%
\usepackage{amsfonts}%
\usepackage{amssymb}
\usepackage{overpic}
\usepackage{subfigure}
\usepackage{rotating}
\usepackage[ruled]{algorithm}
\usepackage{enumerate}
\usepackage{url}
\usepackage{mathtools}
\usepackage{colortbl}
\usepackage{epstopdf}
\usepackage{comment}
\usepackage{natbib}

\usepackage{setspace}
\usepackage[hmargin=1.0in,vmargin=1.0in]{geometry}

\newcommand{\bI}{ {\boldsymbol I} }

\newcommand{\bk}{ {\boldsymbol k} }
\newcommand{\bK}{ {\boldsymbol K} }

\newcommand{\bs}{ {\boldsymbol s} }

\newcommand{\bw}{ {\boldsymbol w} }

\newcommand{\bx}{ {\boldsymbol x} }
\newcommand{\bX}{ {\boldsymbol X} }
\newcommand{\by}{ {\boldsymbol y} }

\newcommand{\bz}{ {\boldsymbol z} }

\newcommand{\bbeta}{ {\boldsymbol \beta} }

\newcommand{\bgamma}{ {\boldsymbol \gamma} }

\newcommand{\bTheta}{ {\boldsymbol \Theta} }

\newcommand{\bzero}{ {\boldsymbol 0} }

\newcommand{\given}{\,|\,}

\DeclarePairedDelimiter\floor{\lfloor}{\rfloor}

\newtheorem{theorem}{Theorem}[section]
\newtheorem{lemma}[theorem]{Lemma}

\newenvironment{proof}[1][Proof]{\begin{trivlist}
\item[\hskip \labelsep {\bfseries #1}]}{\end{trivlist}}

\newcommand{\qed}{\nobreak \ifvmode \relax \else
      \ifdim\lastskip<1.5em \hskip-\lastskip
      \hskip1.5em plus0em minus0.5em \fi \nobreak
      \vrule height0.75em width0.5em depth0.25em\fi}

\title{Large Multi-scale Spatial Kriging Using Tree Shrinkage Priors}
\author{Rajarshi Guhaniyogi and Bruno Sans\'o}

\begin{document}
\maketitle
\begin{abstract}
We develop a multiscale spatial kernel convolution technique with higher order functions to capture fine scale local features and lower order terms to capture large scale features. To achieve parsimony, the coefficients in the multiscale kernel convolution model is assigned a new class of ``Tree shrinkage prior" distributions. Tree shrinkage priors exert increasing shrinkage on the coefficients as resolution grows so as to adapt to the necessary degree of resolution at any sub-domain. Our proposed model has a number of significant features over the existing multi-scale spatial models for big data. In contrast to the existing multiscale approaches, the proposed approach auto-tunes the degree of resolution necessary to model a subregion in the domain, achieves scalability by suitable parallelization of local updating of parameters and
is buttressed by theoretical support.
Excellent empirical performances are illustrated using several simulation experiments and a geostatistical analysis of the sea surface temperature data from the pacific ocean.
\end{abstract}

\section{Introduction}
Ubiquity of spatially indexed datasets in various disciplines \citep{gelfand2010handbook,cressie2015statistics,banerjee2014hierarchical} has motivated researchers to develop variety of methods and models in spatial statistics. Gaussian processes offer a rich modeling framework and are being widely deployed to help researchers comprehend complex spatial phenomena. However, Gaussian process likelihood computations involve matrix factorizations (e.g., Cholesky) and determinant computations for large spatial covariance matrices that have no computationally exploitable structure. This incurs onerous computational burden and is referred to as the ``Big-N'' problem in spatial statistics.

There are, broadly speaking, two different premises for modeling large spatial datasets. One of them is ``sparsity'', while the other is ``dimension-reduction''.
Sparse methods include covariance tapering (see, e.g., \cite{furrer2012covariance,kaufman2008covariance,du2009fixed,shaby2012tapered}), which introduces sparsity in the Gaussian covariance matrix using compactly supported covariance functions. This is effective for fast parameter estimation and interpolation of the response (``kriging"), but is less suited for more general inference on residual or latent processes due to exorbitantly expensive determinant computation of the sparse covariance matrix. An alternative approach introduces sparsity in the inverse of covariance (precision) matrix using conditional independence assumptions or composite likelihoods (e.g., \cite{vecchia1988estimation}; \cite{rue2009approximate}; \cite{stein2004approximating}; \cite{eidsvik2014estimation}; \cite{datta2015hierarchical}; \cite{guinness2016permutation}). In related literature pertaining to computer experiments, localized approximations of Gaussian process models are proposed, see for e.g. \cite{gramacy2015local}, \cite{zhang2016local} and \cite{park2017patchwork}. This literature is overwhelmingly frequentist, less model based and has different goals compared to the Bayesian spatial literature.


Dimension-reduction methods subsume the popular ``low-rank'' models which express the realizations of the Gaussian process as a linear combination of $r$ basis functions (see, e.g., \cite{higdon2002space,stein2007spatial,banerjee2008gaussian,cressie2008fixed,crainiceanu2008bivariate,
finley2009hierarchical,lemos2009spatio}), where $r<<n$. This leads to a flexible class of models, but otherwise possess a low-rank structure that enables likelihood evaluation by solving $r\times r$ linear systems instead of $n\times n$.  The algorithmic cost for model fitting decreases from $O(n^3)$ to $O(nr^2 +r^3)$. However, when $n$ is large, empirical investigations suggest that $r$ must be fairly large to adequately approximate the parent process so that $nr^2$ flops becomes exorbitant. Furthermore, low rank models perform poorly when neighboring observations are strongly correlated and the spatial signal dominates the noise \citep{stein2014limitations}. Improvements on low-rank models with properly designed basis functions (e.g. \cite{guhaniyogi2011adaptive}) have appeared in the recent past. However, these improvements detract from the computational advantages.

Some variants of dimension-reduction methods partition the large spatial data into subsets containing fewer observations, run Gaussian processes in different subsets followed by combining inference from subsets, see e.g. \cite{gramacy2012bayesian}, \cite{guhaniyogi2017}.
Another important class of models aimed at modeling the spatial surface at multiple scales; finer variations at the local scale and overall trend at the global scale. Most of the geophysical processes naturally tend to
have a multiscale character over space that requires statistical methods to allow for potentially complicated multiscale spatial dependence beyond a simple parametric model. However, literature on Bayesian multiscale spatial models for big data is quite insufficient.

Our approach combines the representation of a random field using a multiresolution basis with coefficients modeled using a newly developed \emph{multiscale tree shrinkage prior}. To be more precise, the spatial surface $w(\bs)$ is viewed as the sum of $R$ independent processes
$w(\bs)=\sum_{r=1}^{R} w_r(\bs)$, the $r$-th process $w_r(\bs)$ corresponds to the $r$-th resolution, and
each $w_r(\bs)$ being modeled using a discrete kernel convolution approach,
\begin{align}\label{multi}
w_r(\bs)=\sum_{j=1}^{J(r)}\kappa_{j,r}(\bs)\beta_j^r,
\end{align}
where $\{\kappa_{j,r}(\bs)\}_{j=1}^{J(r)}$, $r\in\{1,...,R\}$, is a set of basis of functions for the $r$th resolution and $\beta_j^r$'s are corresponding coefficients. We use families of radial basis functions of minimal order (see Section 2) kept on a regular grid with increasing resolutions. These radial basis functions have compact support that facilitates significant computational gain. $\beta_j^r$'s play an important role in determining whether a sub-domain requires modeling at the finest resolution
or at coarse resolutions. A new class of \emph{multiscale tree shrinkage prior} is developed to adequately model $\beta_j^r$'s in various sub-domains at different resolutions.

It is noteworthy that a few other important article on multiscale spatial models for big data have already appeared in the literature, see e.g.
\cite{nychka2015multiresolution,katzfuss2016multi,katzfuss2013bayesian} and references therein. We derive a number of very important advantages over the existing literature. Firstly, the most important difference of our approach with \texttt{LatticeKrig} model \citep{nychka2015multiresolution} is that the newly proposed multiscale tree shrinkage prior is equipped to impart increasing shrinkage on basis coefficients as resolution increases. This effectively leads to a continuous analogue to selecting the number of resolutions necessary for modeling a sub-domain. The idea of the tree shrinkage prior is novel in its own right with possible applications anticipated in statistical genomics and neuroscience, for example identifying main effects versus interaction effects in genetic studies. Secondly, unlike \texttt{LatticeKrig}, the proposed multiscale approach incorporates data dependent choice of the kernel width. As a result, similar performance in terms of surface interpolation is available from these models with the proposed multiscale approach using much less number of knots. Thirdly,
the proposed multiscale structure can be naturally embedded in a hierarchical structure to model non-Gaussian data.
Fourthly, this article characterizes the function space of the fitted spatial surface and show asymptotic result on consistency of the posterior distribution of the same. Finally, judicious choice of the compact basis functions and the computational strategy described in Section~\ref{computation} evoke extremely rapid Bayesian estimation that only involves inverting a large number of small matrices in parallel.

The remainder of the manuscript evolves as follows. Section~\ref{MSK} outlines the multiscale kernel convolution model development including the choice of knots, basis functions, basis coefficients and priors on them. Section~\ref{posterior_inference} discusses posterior computation strategies and computation complexities. Theoretical insights on asymptotic properties of the posterior distribution of spatial surface is offered in Section~\ref{theory}. Detailed simulation studies are shown in Section~\ref{simulation}. Section~\ref{SST} details out analysis of a massive sea surface temperature data in pacific ocean. Finally, Section~\ref{conclusion} discusses what the newly developed multiscale model achieves, and proposes a number of future directions to explore.

\section{Multiscale Spatial Kriging}\label{MSK}
\subsection{Kernel convolutions as approximations to Gaussian processes}
Let $\{w(\bs):\bs\in\mathcal{D}\}$ be a spatial field of interest in the continuous domain $\mathcal{D}\subseteq \mathbb{R}^d$, $d\in\mathbb{N}^{+}$. We assume the true spatial process $w(\bs)$ follows a Gaussian process. One may construct a Gaussian process
$w(\bs)$ over $\mathcal{D}$ by convolving a continuous white noise process $u(\bs)$, $\bs\in\mathbb{D}$ with a smoothing kernel $K(\bs,\phi)$ ($\phi$ might be space varying)
so that $w(\bs)=\int K(\bs-\bz,\phi)u(\bz)d\bz$, as proposed by \cite{higdon2002space}. The resulting covariance function for $w(\bs)$ is fully determined
by the kernel $K(\cdot)$ such as
\begin{align}\label{discrete}
\mbox{cov}(w(\bs),w(\bs'))=\int K(\bs-\bz,\phi)K(\bs'-\bz,\phi)d\bz.
\end{align}
A discrete approximation of (\ref{discrete}) is obtained by sampling the convolved processes on a grid. Letting $\bs_1^*,...,\bs_{J}^*$ be a set of knots in $\mathcal{D}$, a discrete approximation of $w(\bs)$ is given by
\begin{align}\label{DCT}
\theta(\bs)=\sum_{j=1}^{J}K(\bs-\bs_j^*,\phi)u_j,
\end{align}
where $u_j$'s are basis coefficients. The $J$ knots are typically placed in a grid in $\mathcal{D}$, though other placements of knots have appeared in the literature. Varying the choice of the kernel functions and coefficients $u_j$, a rich variety of processes emerge from (\ref{DCT}). Following \cite{lemos2009spatio}, we term (\ref{DCT}) as Discrete Process Convolutions (DCT). When $J$ is small DCT provides computationally convenient approximation of the Gaussian process $w(\bs)$. However, Smaller $J$ would greatly reduce approximation accuracy, while moderately large $J$ exacerbates computational burden. The computational challenges cannot be solved by brute-force use of high-performance computing systems, and approximations or simplifying assumptions are necessary. One compelling idea to both reduce computation and increase approximation accuracy may come from using DCT at multiple scales with proper choice of kernel functions and basis coefficients. Next few sections carefully develop a multiscale-DCT model.

\subsection{Partition of domain and choice of knots}\label{domknots}
To define the multiscale-DCT with resolution $R$, we iteratively partition $\mathcal{D}$
up to level $R$. Let at the lowest level, one partitions  $\mathcal{D}$ into $J(1)$ subsets $\mathcal{D}_1,...,\mathcal{D}_{J(1)}$. In the second level, each $\mathcal{D}_i$ undergoes $P$ partitions so that the total number of partitions in the second level is $PJ(1)$. Likewise, let in the $(r-1)$th level the set of partitions can be described as $\{\mathcal{D}_{i_1,..,i_{r-1}}:i_1\in\{1,2,..,J(1)\}, i_2,...,i_{r-1}\in\{1,...,P\}\}$. In the $r$th level each $\mathcal{D}_{i_1,..,i_{r-1}}$
is partitioned into $P$ subsets $\mathcal{D}_{i_1,..,i_{r-1},1}$,..,$\mathcal{D}_{i_1,..,i_{r-1},P}$, so that
\begin{align}
\mathcal{D}_{i_1,..,i_{r-1}}=\bigcup_{i_r=1}^{P}\mathcal{D}_{i_1,..,i_{r-1},i_r},\:\:\mathcal{D}_{i_1,..,i_{r-1},s}\bigcap \mathcal{D}_{i_1,..,i_{r-1},s'}=\phi,\:\:\forall\:s\neq s'.
\end{align}
Therefore, the number of partitions at the $r$th resolution is $J(r)=P^{r-1}J(1)$.
In one dimensional ($d=1$) case, $P=2$, i.e. bisection method is adopted to partition each subset for the next resolution. This naturally implies that the number of partitions at the $r$th level is $J(r)=2^{r-1} J(1)$. In the two dimensional examples, any subset at a resolution is divided into $4$ equal subsets, i.e. $J(r)=4^{r-1}J(1)$. Partitioning of the domain can be envisioned as formation of a tree, with sub-domains $\mathcal{D}_{i_1,...,i_r}$'s as nodes of the tree. Lower and higher resolutions correspond to the upper and lower nodes of this tree. $\mathcal{D}_1$,...,$\mathcal{D}_{J(1)}$ correspond to uppermost nodes of the tree. $P$ branches emerge from each of these nodes leading to
$P^2$ nodes in the second level of the tree and this process continues. Indeed, for any $i_1,...,i_r$, $1\leq r\leq R$, we define $Subtree(\mathcal{D}_{i_1,..,i_{r}})$ by
\begin{align}\label{subtree}
Subtree(\mathcal{D}_{i_1,..,i_{r}})=\{\mathcal{D}_{i_1,..,i_{r}}\}\cup_{j=1}^{R-r-1}\{\mathcal{D}_{i_1,..,i_{r},i_{r+1},..,i_{r+j}}: i_{r+1},...,i_{r+j}\in\{1,...,P\}\}\cup\{\mathcal{D}_{i_1,..,i_{R}}\}.
\end{align}
$Subtree(\mathcal{D}_{i_1,..,i_{r}})$ consists of all sub-domains of $\mathcal{D}_{i_1,..,i_{r}}$ in higher than $r$th resolution, including itself.
Evidently, $Subtree(\mathcal{D}_{i_1,..,i_{R}})=\mathcal{D}_{i_1,..,i_{R}}$. On a similar note, we also define the \emph{father} node of $\mathcal{D}_{i_1,..,i_{r}}$ as the node $\mathcal{D}_{i_1,..,i_{r-1}}$.

Defining multiscale-DCT also requires choosing a set of knot points at every level. We place knots in the center of every partition at a resolution. To be more precise, the knots $\bs_1^1,...,\bs_{J(1)}^1$ in the first level is placed at the centers of $\mathcal{D}_1,...,\mathcal{D}_{J(1)}$. Likewise, knots $\bs_1^r,...,\bs_{J(r)}^r$ are kept at the centers of the partitions at the $r$th level. Therefore, there is a one-one correspondence between the set of knots and the set of partitions of $\mathcal{D}$. Henceforth, we will interchangeably use $Subtree$ and $Father$ of a sub-domain with $Subtree$ and $Father$ of the knot that resides at the midpoint of that sub-domain, e.g. if $\bs_j^r\in\mathcal{D}_{i_1,...,i_r}$, then $Subtree(\bs_j^r)$ and $Father(\bs_j^r)$ are synonymous with $Subtree(\mathcal{D}_{i_1,...,i_r})$ and $Father(\mathcal{D}_{i_1,...,i_r})$ respectively.
It is be noted that the indexing set of knots is a bit different from the indexing set of partitions and they require reconciliation.
The $j$th knot at the $r$th resolution $\bs_j^r$, $j=1,..,J(r)$, belongs to the sub-domain $\mathcal{D}_{i_1,...,i_r}$ if $j=\sum_{l=1}^{r-1}(i_l-1)P^{r-l}+i_r$. With this notation, $\bs_k^{r-1}$ is the father node of $\bs_j^r$ iff
$k=\sum_{l=1}^{r-2}(i_l-1)P^{r-l}+i_{r-1}$, i.e. $k=\floor{\frac{j-1}{P}}+1$, where $\floor{x}$ is the greatest integer less than $x$.
\begin{figure}[h]
\begin{center}
\subfigure[knot placement: one dimension]{\includegraphics[width=8cm]{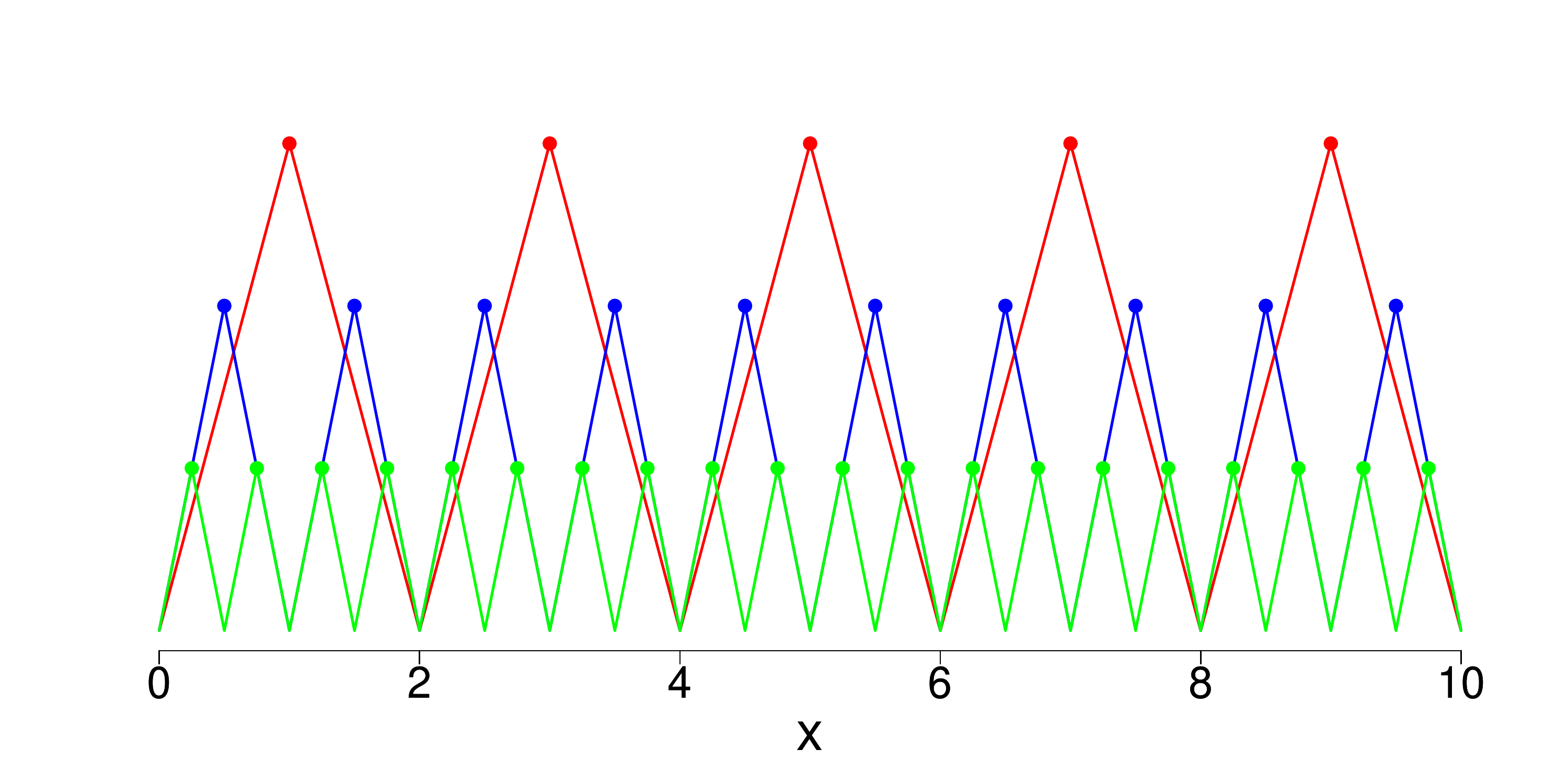}}\\
\subfigure[knot placement: resolution 1 (2d)]{\includegraphics[width=6cm]{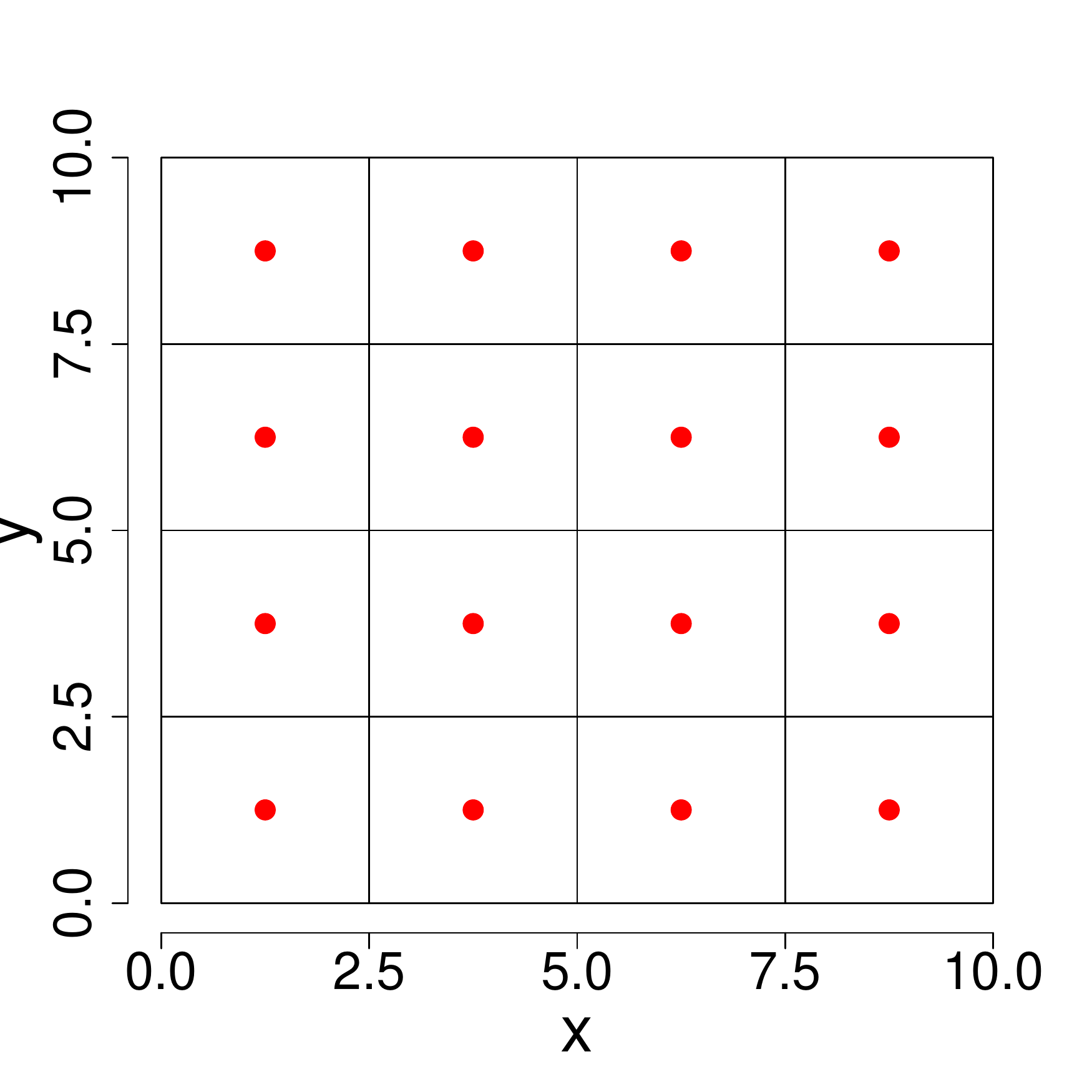}}
\subfigure[knot placement: resolution 2 (2d)]{\includegraphics[width=6cm]{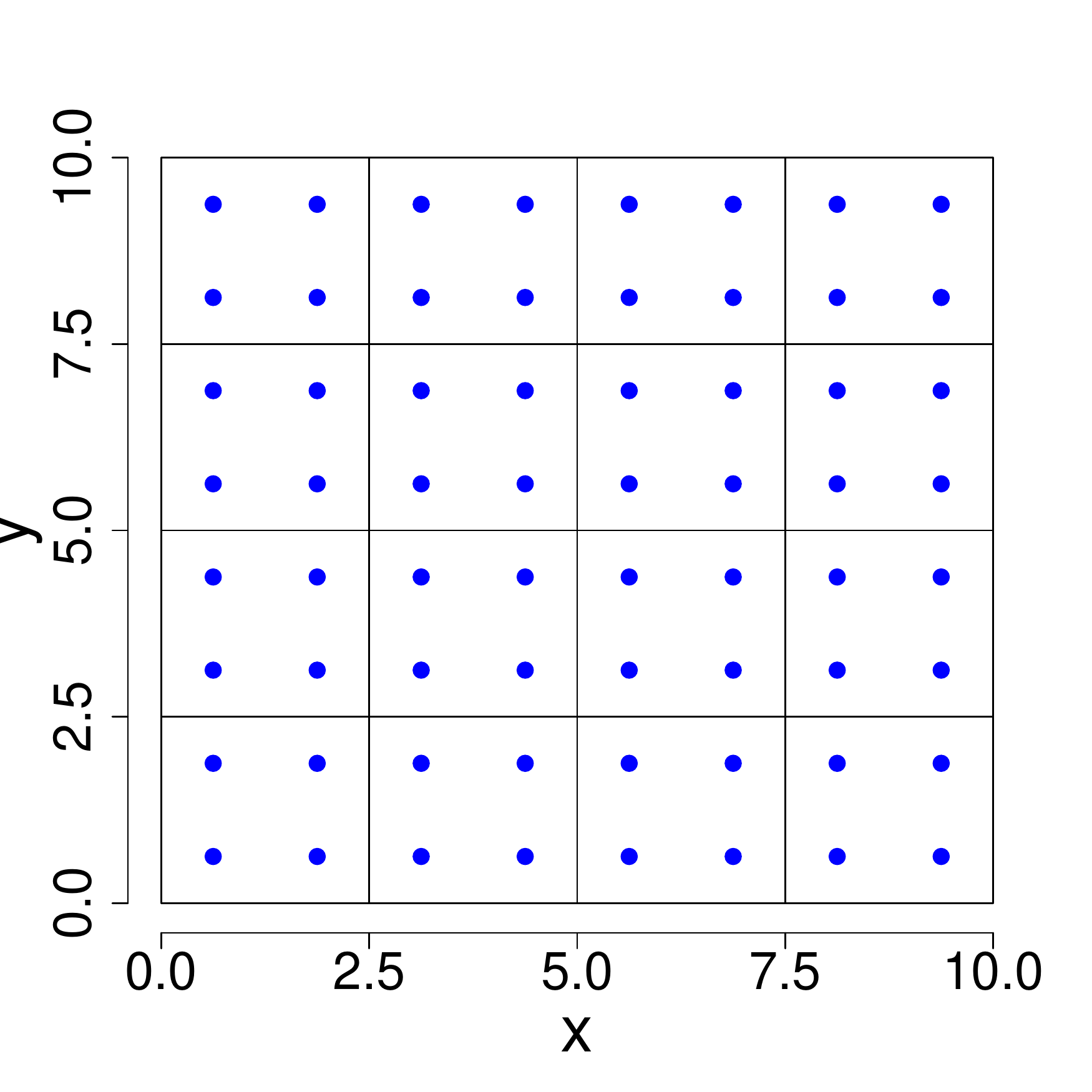}}
\end{center}
\caption{(a) placement of knots in one dimension for three resolutions. Knots in resolution 1 are shown in the upper level. Middle and lower level show knots in resolutions 2 and 3 respectively. For visualization in one dimension, we fix $R=3$, $J(1)=5$. (b) Shows placement of knots in resolution 1 for two dimensions; (c) shows placement of knots in resolution 2 for two dimensions. For better visualization, we keep $R=2$, $J(1)=16$ in two dimensions.}\label{knots}
\end{figure}

To give examples of domain partitioning and knots, consider the one dimensional case with the spatial domain of interest $[h_1,h_2]$. In the first resolution, the domain is partitioned into $J(1)$ intervals of equal length, i.e. each interval is of length $\delta=(h_2-h_1)/J(1)$. The knots in the first resolution are placed at the midpoint of each interval so that the spacing between two successive knots is always $\delta$. At the second resolution, each interval is partitioned into equal length intervals with knots in resolution 2 placed right at the midpoints of these intervals. Therefore,
the space between two successive knots is reduced to $\delta/2$. This process continues iteratively leading to the space between two successive knots at resolution $r$ as $\delta/2^{r-1}$. For $d=2$, let the domain of interest be $[h_1,h_2]\times[h_3,h_4]$. The first resolution divides the area into $h_x\times h_y$ equi-dimensional rectangles with knots placed in the center of each rectangle. Therefore, the number of knots in the first resolution is $J(1)=h_x\times h_y$. In resolution 2, every rectangle in the first resolution is divided into four congruent rectangles. Knots in the second resolution are placed at the centers of these new rectangles. Clearly, the distance between two horizontally adjacent knots is $(h_2-h_1)/h_x$ and two vertically adjacent knots is $(h_4-h_3)/h_y$, in the first resolution. At the $r$th resolution these distances decrease to $2^{-r+1}(h_2-h_1)/h_x$ and $2^{-r+1}(h_4-h_3)/h_y$ respectively.

Figure~\ref{knots} shows the domain partitioning and the set of knots for both one dimensional and two dimensional applications. For the visual illustration, we restrict to $R=3$, $h_1=0,h_2=10,h_3=0,h_4=10$, $J(1)=5$ in one dimension. For two dimensions, we restrict $R=2$, $h_x=4,h_y=4$, $J(1)=16$ for better visualization. Hereon, we fix the template of domain partitions and the placement and number of knots in each partition. We offer some more discussion on them in the conclusion and future work section.

\subsection{Multiscale spatial process with radial basis functions}
We model the spatial effects by a multi-scale DCT with $R$ resolutions, $r$th resolution being modeled by a DCT with kernel $K(\cdot,\cdot,\phi_r)$, knots $\bs_1^r,...,\bs_{J(r)}^r$ and coefficients $\beta_1^r,...,\beta_{J(r)}^r$, $r=1,..,R$. To elaborate
it further, the spatial surface $w(\bs)$ is written as $w(\bs)=\sum_{r=1}^{R}w_r(\bs)$, such that
\begin{align}\label{rthscale}
w_r(\bs)=\sum_{j=1}^{J(r)}K(\bs,\bs_j^r,\phi_r)\beta_j^r.
\end{align}
Contrary to the one scale DCT, multiscale DCT captures spatial variability at multiple scales. The lower resolutions capture variability at
large distances while the finer local level variabilities are captured by higher resolutions. Intuitively, the implication is that
one needs more basis functions in one scale DCT compared to multiscale DCT to draw similar inference. We will formally discuss this issue in simulation studies.

The choice of the kernel function $K(\cdot,\cdot,\phi_r)$ is crucial for estimating the spatial variability at multiple scales. In the context of
the ordinary one resolution kernel convolution literature, \cite{lemos2009spatio,cressie2008fixed} proposed to use Bezier kernels which are continuous but not differentiable. In the multiscale literature, \cite{nychka2015multiresolution} uses a Wendland kernel that is four times continuously differentiable. Let $\kappa$ be a Wendland polynomial functions \citep{wendland2004scattered}, supported on [0,1], having the form
\begin{align}\label{compact}
\kappa(z)=(1-z)_{+}^{l+1}(1+(l+1)z),
\end{align}
where $(1-z)_{+}=(1-z)$ if $0<z<1$ and $=0$ otherwise and $l=\floor*{d/2}+2$. For our proposed
approach, we choose kernel function $K$ defined as
\begin{align}\label{eq:kernel}
K(\bs,\bs_j^r,\phi_r)=\kappa\left(\frac{||\bs-\bs_j^r||}{\phi_r}\right)=\left(1-\frac{||\bs-\bs_j^r||}{\phi_r}\right)_{+}^{l+1}\left[1+(l+1)\frac{||\bs-\bs_j^r||}{\phi_r}\right].
\end{align}
Geometrically, the kernel function consists of bumps centered at the node points with interpolation of the spatial surface at $\bs$ in the $r$th resolution is governed by knots located in $B_{\phi_r}(\bs)$, where $B_{\nu}(\bs)$ is the Euclidean ball of radius $\nu$ around $\bs$.
Section~\ref{computation} describes computational advantages derived from the compact support of this kernel.

Note that $\kappa$ is a Wendland polynomial function supported on $[0,1]$.
\cite{wendland2004scattered} ensures that $\kappa$ is positive definite
which seems to be an attractive feature when this function is used for interpolation.
Further, $\kappa(z)\in C^2$ and it is the positive definite compactly supported
polynomial of minimal degree for a given dimension $d$ that possesses continuous
derivatives up to second order. Theorem~\ref{RKHS} characterizes the space of functions of the form $w_r(\bs)$ spanned by the basis functions $K(\bs,\cdot,\phi_r)$. Proof of the Theorem~\ref{RKHS} is given in the Appendix.
\begin{theorem}\label{RKHS}
Consider the Reproducing Kernel Hilbert Space (RKHS) of the space of functions
\begin{align*}
\mathcal{H}_r=Span\left\{K(\bs,\cdot,\phi_r)\right\}
\end{align*}
spanned by the kernel at the $r$th resolution. Then $\mathcal{H}_r=\mathcal{S}^{d/2+3/2}(\mathcal{R}^d)$, where
\begin{align*}
\mathcal{S}^{d/2+3/2}(\mathcal{R}^d)=\{f\in L_2(\mathcal{R}^d)\cap C(\mathcal{R}^d): \hat{f}(\cdot)(1+||\cdot||^2)^{(d+3)/4}\in L_2(\mathcal{R}^d)\},
\end{align*}
is the Sobolev space of order $d/2+3/2$, and $\hat{f}(\cdot)$ is the Fourier transform of
$f(\cdot)$.
\end{theorem}
\textbf{Remark:} The result establishes that the sample paths of $w_r(\bs)$ belong to $\mathcal{S}^{d/2+3/2}(\mathcal{R}^d)$. Roughly speaking, for integer $\zeta$, $\mathcal{S}^{\zeta}$ contains functions whose derivatives upto $\zeta$-th order are continuously differentiable. Thus, $w_r(\bs)$ constructed in this way ought to provide continuously differentiable realizations of the spatial surface a priori.

The choice of the scale parameter $\phi_r$ for the $r$th resolution follows from several considerations. Since the kernels in lower resolutions
are meant to capture long range variabilities, one naturally imposes the constraint
\begin{align}\label{phieq}
\phi_1>\phi_2>\cdots>\phi_R>0.
\end{align}
Secondly, given $\beta_j^r$, $j=1,...,J(r)$, $r=1,...,R$, $\phi_r$ determines the set of
knots in the neighborhood of $\bs$ for interpolating the spatial surface at $\bs$. One
could possibly keep $\phi_r$ as a parameter and update them using MCMC. Our detailed
investigation reveals that such an approach adds unnecessary computational burden with no
apparent inferential advantage. Therefore, this article fixes
$\phi_r=\eta||\bs_j^r-\bs_{j-1}^r||$, $\eta>0$. $\eta$ is a tuning parameter that
determines the computational advantage vis a vis long range spatial dependence between
observations. Since $\eta$ is not a parameter of interest, the present article does not attempt to make full scale Bayesian inference on $\eta$. Rather at each step of the MCMC iteration posterior likelihood is maximized over a grid of $\eta$. We further elaborate it in Section 3.1.

\subsection{Mutiscale spatial regression model}
Our proposed Mutiscale spatial model typically assume, at location $\bs\in\mathcal{D}$, a response variable $y(\bs)\in\mathcal{R}$ along with a $p\times 1 $ vector of spatially referenced predictors $\bx(\bs)$ which are associated through a spatial regression model as
\begin{align}\label{multiscale}
y(\bs)=\bx(\bs)'\bgamma+\sum_{r=1}^{R}\sum_{j=1}^{J(r)}K(\bs,\bs_j^r,\phi_r)\beta_j^r+\epsilon(\bs),\:\:\epsilon(\bs)\sim N(0,\sigma^2),
\end{align}
where $\bgamma$ is the $p\times 1$ vector of regression coefficient. The medium and short
range spatial variability of $y(\bs)$ is determined by the multiscale DCT term, while
$\epsilon(\bs)$ adds a jitter that corresponds to unexplained micro-scale variability or
measurement error, with $\sigma^2$ as the error variance.

\subsection{Mutiscale shrinkage prior on $\beta_j^r$}
Once the model formulation is complete, attention turns to assigning prior distributions on $\beta_j^r,\bgamma,\sigma^2$. While the prior specification on $\bgamma$ and $\sigma^2$ is straightforward, specifically $\bgamma$ is assigned a noninformative prior and $\sigma^2\sim IG(c,d)$, constructing prior
distribution on $\beta_j^r$ requires a bit of reflection. Note that the local variability
within the spatial domain varies in relation to the sub-domains. There are regions within the domain which exhibit small scale spatial variability, thereby required to be modeled by the higher resolutions. On the other hand, spatial variability is less prominent in some regions, which practically do not require higher resolutions for modeling. Mathematically, this amounts to setting $\beta_j^r=0$ corresponding to the knots $s_j^r$ located in the latter regions. It is also natural to assume that if $r$th resolution deemed unnecessary to model the surface in a sub-domain, any $l$th resolution for $l>r$ should be unnecessary too to model the same subregion. For $\bs_j^r\in\mathcal{D}_{i_1,...,i_r}$, define,
\begin{align*}
\mathcal{B}_{j,r}^{Subtree}=\left\{\beta_{k}^l: l\geq r, \bs_{k}^l\in Subtree(\mathcal{D}_{i_1,...,i_r})\right\}.
\end{align*}
Thus $\mathcal{B}_{j,r}^{Subtree}$ is the set of coefficients corresponding to basis functions centered at knots in $Subtree(\mathcal{D}_{i_1,...,i_r})$. As per our discussion, if modeling small scale spatial variation in the subregion $\mathcal{D}_{i_1,...,i_r}$ does not require resolution $r$, any higher resolution deem unnecessary too.
This leads to condition C\\
\textbf{Condition C:} $\beta_j^r=0$ implies $\beta_k^l=0$, where $\beta_k^l\in\mathcal{B}_{j,r}^{Subtree}$.\\
The problem of estimating $\beta_j^r$'s finds equivalence in the variable selection literature in high dimensional regression. The goal in variable selection literature lies in identifying predictors not related to the response, equivalently the predictors having zero coefficients. A rich variety of methods have been proposed ranging from penalized optimization methods, such as Lasso \citep{tibshirani1996regression} and elastic net \citep{zou2005regularization}, to Bayesian variable selection or shrinkage methods. Penalized optimization is computationally efficient in identifying unimportant predictors even in the presence of large number of predictors, but suffers due to their inability to produce accurate characterization of uncertainty. Besides, penalized optimization results are highly sensitive to the adhoc choice of tuning parameters. The Bayesian approach is attractive in its probabilistic characterization of uncertainty for regression coefficients in high dimensions and for the resulting predictions.
The most popular artifact employed in the Bayesian high dimensional regression for variable selection is the wide class of
spike and slab prior distributions on predictor coefficients. The widespread usage of spike and slab priors is a consequence of its easy interpretability and relatively simple computation. It is be noted that \textbf{Condition C} hinders standard application of a spike and slab prior on $\beta_j^r$. One can possibly build a new class of spike-and-slab prior over the traditional spike-slab prior for selective predictor inclusion \citep{george1993variable,geweke2004getting,clyde1996prediction} respecting the constraints imposed by \textbf{Condition C}. However,
spike and slab prior faces notorious mixing issues and consequently inaccurate inference for more than a few hundred predictors. This has motivated us to derive a continuous shrinkage prior that does not set $\beta_j^r=0$, but imposes
a stochastic ordering between $\beta_j^r$ along resolutions a priori. The literature on high dimensional regression have consistently found that
shrinkage-promoting priors are more effective in practice on real (natural) data than exact-sparsity-promoting models, like the spike-slab  setup.

An impressive variety of Bayesian shrinkage priors for ordinary high dimensional regressions with
scalar/vector response on high dimensional vector predictors has been proposed in recent
times, see for example \cite{armagan2013generalized,hans2009bayesian,park2008bayesian,polson2012local,carvalho2009handling} and references therein. Shrinkage priors are based on the
principle of artfully shrinking predictor coefficients of unimportant predictors to zero, while
maintaining proper estimation and uncertainty of the important predictor coefficients. However,
the literature on shrinkage priors that impose increasing shrinkage along resolutions is
quite insufficient. This article introduces a \emph{multiscale tree shrinkage
prior} to achieve this objective. It proposes
\begin{align}\label{multsh}
\beta_j^r &\sim N(0,\alpha_j^r)\nonumber\\
\alpha_j^1=\delta_1^{-1}, \alpha_j^2 &=\delta_1^{-1}\delta_{j,2}^{-1},
\alpha_j^r=\alpha_{\floor{\frac{j-1}{P}}+1}^{r-1}\delta_{j,r}^{-1}\nonumber\\
\delta_1\sim Gamma(2,1), &\delta_{j,r}\sim Gamma(c,1), c>2.
\end{align}
$\delta_{j,r}^{-1}$'s are stochastically smaller than $1$ implying increasing shrinkage apriori along a branch. In fact,
$E[\beta_j^r]=0$ and $Var[\beta_j^r]=\frac{1}{(c-1)^{r-1}}\rightarrow 0$, as $r\rightarrow\infty$, apriori. Thus the prior distribution imposes strong apriori belief of having a parsimonious model with small number of resolutions. The proposed prior offers easy posterior updating with closed form conditional posterior distributions for all the parameters, as is discussed in the next section.

\section{Posterior computation and inference}\label{posterior_inference}
This section describes posterior computation and inference for multiscale DCT. The main
task for inference remains that of obtaining the
posterior distribution of the unknown coefficients $\beta_j^r$ and $\delta_{j,r}$ $j=1,...,J(r)$ and $r=1,...,R$, $\bgamma$ and $\sigma^2$.
The formulation of multiscale DCT is simple, so that all the parameters allow simple Gibbs sampling updates. Once posterior distributions of the parameters are available, they are employed to estimate interpolation of the residual surface and perform spatial predictions. By crucially exploiting the conditional independence among several parameters and multi-resolution structure of the problem, we obtain inference with excellent time and memory complexity (Section~\ref{computation} and \ref{simulation}), can take full advantage of distributed-memory systems with a large number of nodes (Section~\ref{computation}), and is thus scalable to large spatial datasets.

\subsection{Posterior computation}
We proceed to do parametric inference with data $(y(\bs_i),\bx(\bs_i))_{i=1}^n$ at locations $\mathcal{S}=\{\bs_1,...,\bs_n\}$. Stacking responses and predictors across locations we obtain, $\by=(y(\bs_1),...,y(\bs_n))'$, $\bX=[\bx(\bs_1):\cdots :\bx(\bs_n)]'$. Let $\bK$ be an $n\times (J(1)+\cdots+J(R))$ matrix whose $i$th row is
given by $(K(\bs_i,\bs_1^1,\phi_1),...,K(\bs_i,\bs_R^{J(R)},\phi_R))'$. Further assume $\bbeta^r=(\beta_1^r,...,\beta_{J(r)}^r)'$ and $\bbeta=(\bbeta^1,...,\bbeta^R)'$, $y_{i,r,j}=y_i-\sum_{(k_1,k_2)\neq (j,r)}K(\bs_i,\bs_{k_1}^{k_2},\phi_{k_2})\beta_{k_1}^{k_2}$, $\by_{r,j}=(y_{1,r,j},...y_{n,r,j})'$, $\bK_{r,j}=(K(\bs_1,\bs_j^r,\phi_r),...,K(\bs_n,\bs_j^r,\phi_r))'$. The full conditional distributions of $\bgamma$, $\sigma^2$, $\beta_j^r$ and $\delta_{j,r}$ are readily available in closed form and are given by
\begin{itemize}
\item $\bgamma|-\sim N((\bX'\bX)^{-1}\bX'(\by-\bK\bbeta),\sigma^2(\bX'\bX)^{-1})$
\item $\sigma^2|-\sim IG\left(\frac{n}{2}+c,d+\frac{1}{2}||\by-\bX\bgamma-\bK\bbeta||^2\right)$
\item $\beta_j^r|-\sim N\left(\frac{\frac{\bK_{r,j}'\by_{r,j}}{\sigma^2}}{\frac{1}{\alpha_j^r}+\frac{\bK_{r,j}'\bK_{r,j}}{\sigma^2}},\frac{1}{\frac{1}{\alpha_j^r}+\frac{\bK_{r,j}'\bK_{r,j}}{\sigma^2}}\right)$.
\item Recall the definition of father node in Section~\ref{domknots}. Additionally define $father^2(\bs_j^r)$ as the father node of the father node of $\bs_j^r$. Similarly, $father^3,...,father^R$ node are defined. Let $\alpha_{j,r,-1}=\prod\limits_{l=r}^R\delta_{father^{r+1-l}(\bs_j^l),l-1}$ and $\alpha_{j,1,-1}=1$. Then\\
 $\delta_1|-\sim Gamma\left(1+\frac{J(1)+\cdots+J(R)}{2},1+\frac{1}{2}\sum_{r=1}^{R}\sum_{j=1}^{J(r)}
    \left[(\beta_{j}^{r})^2/\alpha_{j,r,-1}\right]
    \right)$
\item Let $\alpha_{k,l,-r,-j}=\delta_1\prod\limits_{h=l}^{r+2}\delta_{father^{l+1-h}(\bs_k^h),h-1}\prod\limits_{h=r}^{2}\delta_{father^{r+1-h}(\bs_j^h),h-1}$,
    $\alpha_{j,r,-r,-j}=1$. Then\\
    $\delta_j^r|-\sim Gamma\left(c+\frac{\#\bbeta_{j,r}^{Subtree}}{2},1+\frac{1}{2}\sum\limits_{l\geq r,\bs_k^r\in Subtree(\bs_j^r)}(\beta_k^l)^2/\alpha_{k,l,-r,-j}\right)$ for $r>1$.
\item Finally at each iteration, joint posterior distribution is maximized over a discrete grid of $\eta$ values, $\eta\in \{1,...,h_{\eta}\}$, $h_{\eta}$ is an integer. In all simulation studies and in the real data analysis, we never found the maximization of the posterior over $\eta$ to occur for $\eta$ values more than $5$. Thus, we fix $h_{\eta}=5$ for all empirical investigations.
\end{itemize}
\subsection{Distributed computation and computational complexities}\label{computation}
An important advantage of the multiscale DCT is that it facilitates distributed computation
with little communication overhead at a large number of nodes, each only dealing with a small
subset of the data. This section describes the distributed computing algorithm as well as the associated
computation complexity.

To begin with, one must acknowledge that posterior updating of $\bgamma, \sigma^2,\delta_{j,r},\delta_1$ can be
carried out rapidly without having to store the entire data in one processor. The main computational difficulty
comes from updating of $\bbeta$. Single updating of $\beta_j^r$ introduces too much autocorrelation, while
joint updating of $\bbeta$ requires inverting $(\sum_{r=1}^{R}J(r))\times (\sum_{r=1}^{R}J(r))$ matrix which is
infeasible. The use of compactly supported basis functions offers an excellent solution by carefully exploiting conditional independence between blocks of $\bbeta$.
For $m=1,...,J(1)$, define the \emph{neighborhood function} $\mathcal{N}(m)$ of $m$ by
\begin{align*}
\mathcal{N}(m)=\left\{j:||\bs_j^1-\bs_m^1||<2\eta\right\}
\end{align*}
Similarly, the \emph{neighborhood data function} is defined as
\begin{align*}
\mathcal{N}_{D}(m)=\left\{j:||\bs_j^1-\bs_m^1||<\eta\right\}
\end{align*}
Let $\bbeta_{j,r}^{Subtree}$ be a vector composed of all elements in $\mathcal{B}_{j,r}^{Subtree}$. Clearly,
$\bbeta=(\bbeta_{1,1}^{Subtree},...,\bbeta_{1,J(1)}^{Subtree})'$.
Exploiting properties of the compactly supported basis functions, one obtains
\begin{align*}
\bbeta_{s,1}^{Subtree}|- \stackrel{\mathcal{L}}{=} \bbeta_{s,1}^{Subtree}|\by_{\mathcal{N}_D(s)},
                       \bbeta_{\mathcal{N}(s),1}^{Subtree}, s=1,...,J(1).
\end{align*}
Here $\mathcal{N}_D(s)=\{i_1,...,i_l\}$ will imply that $\by_{\{i_1,...,i_l\}}$ and $\bX_{\{i_1,...,i_l\}}$ are the data responses and corresponding predictors in the domain $\mathcal{D}_{i_1}\cup\cdots\cup\mathcal{D}_{i_l}$.

\begin{enumerate}[a.]
  \begin{algorithm}[t]
    \caption{Distributed computing of the posterior distribution of $\bbeta,\bgamma,\sigma^2,\delta_{j,r}$} \label{post_comp_al}
    {\scriptsize
\item \textbf{No. of nodes used:} Use $J(1)$ nodes for computation.
\item \textbf{MCMC initialization:} Initialize all parameters.
\item At the $t$th iteration, MCMC iterates are given by $(\bbeta_{m,1}^{Subtree})^{(t)}$, $m=1,...,J(1)$, $\sigma^{2(t)},\bgamma^{(t)}$, $\delta_{j,r}^{(t)}$, $j=2,...,J(r); r=1,...,R$ and $\delta_1^{(t)}$.
\item Maximize posterior likelihood w.r.t. $\eta\in\{1,...,h_{\eta}\}$. Compute $(\by_{\mathcal{N}_D(m)},\bX_{\mathcal{N}_D(m)})$ according to the maximized $\eta$. At the $t$th iteration store $(\by_{\mathcal{N}_D(m)},\bX_{\mathcal{N}_D(m)})$ in the $m$th node.
\item \texttt{For} $m=1:J(1)$ in parallel in $J(1)$ different nodes
\begin{enumerate}[i.]
\item $(t+1)$ iterate of $(\bbeta_{m,1}^{Subtree})^{(t+1)}$ is obtained by drawing from $\bbeta_{1,m}^{Subtree}|(\bbeta_{\mathcal{N}(m),1}^{Subtree})^{(t)}$.
\end{enumerate}
\item \texttt{For} $m=1:J(1)$ in parallel in $J(1)$ different nodes
\begin{enumerate}[i.]
\item Calculate $\bX_m'\bX_m$, $\by_m-\bK_m\bbeta$, where $\bK_m=(K(\bs,\bs_1^1,\phi_1),...,K(\bs,\bs_{J(R)}^R,\phi_R))$, $\bs\in\mathcal{D}_m$.
\end{enumerate}
\item Use the fact that $\sum_{m=1}^{J(1)}\bX_m'\bX_m=\bX'\bX$ and $\by-\bK\bbeta=(\by_1-\bK_1\bbeta,...,\by_{J(1)}-\bK_{J(1)}\bbeta)'$ to update from the full condition of $\bgamma$.
\item Update $\delta_{j,r}^{(t+1)}$ and $\delta_1^{(t+1)}$ at the $(t+1)$th iteration.
 }%
  \end{algorithm}
\end{enumerate}
Algorithm~\ref{post_comp_al} describes details of the computation strategy we adopt. As per Algorithm~\ref{post_comp_al}, the computation involves $J(1)$ nodes with $m$th node storing $\{\by_{\mathcal{N}_D(m)},\bX_{\mathcal{N}_D(m)}\}$ and executing posterior updates of $\bbeta_{m,1}^{Subtree}$. The main computation cost involved in the $m$th node is in computing Cholesky decomposition of a $dim(\bbeta_{\mathcal{N}(m),1}^{Subtree})\times dim(\bbeta_{1,\mathcal{N}(m)}^{Subtree})$ and multiplying a $dim(\mathcal{N}_D(m))\times(\sum_{r=1}^R J(r))$ matrix with a vector of dimension $(\sum_{r=1}^R J(r))$. They incur computation
complexities of $O(dim(\mathcal{N}(m))^3)$ and $O(dim(\mathcal{N}_D(m))\sum_{r=1}^R J(r))$ respectively.
Since $dim(\bbeta_{1,\mathcal{N}(m)}^{Subtree})=((2d)^R-1)/(2d-1)$, the computation time for the former is low. Choosing $J(1)$ large enough one can reduce the computation time for the latter as well. The storage complexity is also dominated by $dim(\mathcal{N}_D(m))$.

\begin{table}[h]
\centering
\begin{tabular}{cc}
	\hline
Time(multiple processor) & Time (single processor) \\
	\hline
&\\
$\left[(2\eta-1)\frac{n}{J(1)}\right]\left[\sum_{r=1}^{R}J(r)\right]+\frac{(2d)^r-1}{2d-1}$ &  $J(1)\left[(2\eta-1)\frac{n}{J(1)}\right]\left[\sum_{r=1}^{R}J(r)\right]$ \\
& \\
	\hline
\end{tabular}
\caption{Dominant terms in calculating time complexities for multiscale DCT with single and multiple processors.}
\label{memtime}
\end{table}


\subsection{Prediction and residual surface interpolation}
 Let $\bs_0$ be any location in the domain, where we seek to predict
$y(\bs_0)$, based on a given vector of predictors $\bx(\bs_0)'$.
The spatial prediction at $\bs_0$ proceeds from the
posterior predictive distribution
\begin{align}\label{Eq: Posterior_Predictive_Y}
 p(\by(\bs_0)\given\by) = \int p(\by(\bs_0)\given \by,\bTheta)p(\bTheta\given\by)\: d\bTheta,
\end{align}
using composition sampling, where $\bTheta=(\sigma^2,\bgamma,(\beta_j^r)_{j,r=1}^{J(r),R},(\delta_{j,r})_{j,r=1}^{J(r),R})$.
For each $\{\bTheta^{(l)}\}$,
$l=1,2,\ldots,L$, obtained from the posterior distribution
$p(\bTheta\given\by)$, draw $\by(\bs_0)^{(l)}$ from
$p(\by(\bs_0)\given \bTheta^{(l)})$. The resulting
$\by(\bs_0)^{(l)},\, l=1,2,\ldots,L$ are samples from (\ref{Eq:
Posterior_Predictive_Y}). This is especially simple for multiscale DCT as $p(\by(\bs_0)\given \bTheta)$ turns out
to be a normal distribution.

For multiscale DCT, full Bayesian inference on the residual spatial surface at any unobserved location $\bs_0$ is trivially obtained.
For each posterior sample $\{\bTheta^{(l)}\}$, $l=1,2,\ldots,L$, compute
$w(\bs_0)_{l}=\sum_{r=1}^{R}\sum_{j=1}^{J(r)}K(\bs_0-\bs_j^r,\phi_r)(\beta_j^r)^{(l)}$.
$w(\bs_0)_l$ are samples from the posterior distribution of the residual process. Surface interpolation
is straightforward hereafter.

\section{Theoretical properties}\label{theory}
We establish convergence results for multiscale DCT regression model \eqref{multiscale} under the simplifying assumptions\footnote{Simplifying assumption is merely to ease notation and calculations; all results generalize in a straightforward manner.} that the predictor coefficient $\bgamma=(0, \dots, 0)$.

Define two metrics in the function space given by
\begin{align*}
||w||_{\infty} &=\sup\limits_{\bs\in\mathcal{D}}|w(\bs)|,\\
||w||_{\zeta} &=\max\limits_{k\leq\floor{\zeta}}\sup\limits_{\bs\in\mathcal{D}}|D^k w(\bs)|+\max\limits_{\tilde{k}\leq\floor{\zeta}}\sup\limits_{\bs,\bs'\in\mathcal{D}}\frac{|D^k w(\bs)-D^k w(\bs')}{||\bs-\bs'||^{\zeta-\floor{\zeta}}},
\end{align*}
where $D^k=\frac{\delta^{k_1+k_2}}{\delta s_1^{k_1}\delta s_2^{k_2}}$, for $k_1,k_2\in\mathbb{N}$ and $\bs=(s_1,s_2)'$.
Further define the sets
\begin{align*}
\Theta_{\zeta} &=\left\{w(\bs):w(\bs)=\sum_{r=1}^{R}\sum_{j=1}^{J(r)}K(\bs,\bs_j^r,\phi_r)\beta_j^r, R\in\mathbb{N},\bs_j^r\in\mathcal{R}^2,
                          \beta_j^r\in\mathcal{R},||w||_{\zeta}<\infty, \right\}\\
\Theta_{\zeta}^{n} &=\left\{w\in\Theta_{\zeta}: ||w||_{\zeta}<n^{\alpha},\alpha\in (1/2,1]\right\}\\
\Theta_{\zeta,c} &= \mbox{Closure under}\: ||\cdot||_{\infty}\:\mbox{of}\:\Theta_{\zeta}\\
\mathcal{B}_{\epsilon,n} &= \left\{w\in\Theta_{\zeta}^{n}:\frac{1}{n}\sum_{i=1}^{n}|w(\bs_i)-w_0(\bs_i)|<\epsilon,\left|\frac{\sigma}{\sigma_0}-1\right|\right\}.
\end{align*}

\begin{theorem}\label{consistency}
Let $\mathcal{P}_{w_0,\sigma_0^2}$ denotes the true data generating joint distribution of $\{y_i\}$. Assume
\begin{enumerate}[(a)]
\item $\mathcal{D}$ is compact.
\item $K(\cdot,\cdot,\phi_r)$ is continuous.
\item $w_0\in\Theta_{\zeta,c}$, $||w_0||_{\zeta}<\infty$, for some $\zeta$.
\end{enumerate}
Then for any
$(w_0,\sigma_0^2)\in\Theta_{\zeta,c}\times\mathcal{R}^{+}$ and for any $\epsilon>0$,
\begin{align*}
\lim\limits_{n\rightarrow\infty} \Pi((w,\sigma^2)\in\mathcal{B}_{\epsilon,n}|y_1,...,y_n)=0
\end{align*}
almost surely under $\mathcal{P}_{w_0,\sigma_0^2}$.
\end{theorem}
Theorem~\ref{consistency} establishes consistency of estimating the data generating surface $w_0$ and the true error variance $\sigma_0^2$. The proof proceeds along the same line of arguments outlined in \cite{choi2007posterior}, \cite{pillai2008levy} and is provided in the Appendix.

\section{Simulation studies}\label{simulation}
\noindent In this section, we use synthetic datasets to assess model performance with regard
to interpolating unobserved residual spatial surface and predicting at new locations. To begin with, we present a one dimensional simulation
experiment on a large dataset. The one dimensional simulation experiment helps to build
intuition on how different resolutions capture large and small scale variabilities,
including the advantage of choosing the tree shrinkage prior. Once computational and
inferential aspects of multiscale DCT are discussed in one dimension, we provide a more realistic two dimensional example where computation time and performance of
multiscale DCT will be compared with state-of-the-art and popular spatial models for big data. A non-distributed implementation of the methods are carried out in \texttt{R} version 3.3.1 on a 16-core machine (Intel Xeon 2.90GHz)
with 64GB RAM.
\subsection{One dimensional Example}
For the one dimensional example, we simulated a dataset of size $n=20,000$ from the model with an intercept, a predictor and a spatial function $w_0(s)$ in $[0,10]$ given by
\begin{align}
w_0(s)=\left\{\begin{array}{l}
\mbox{sin}(2\pi s)s,\:\mbox{if}\:0\leq s<2\\
|\mbox{sin}(s-3)|^3,\:\mbox{if}\:2\leq s<4\\
5|s-5|,\:\mbox{if}\:4\leq s<6\\
 \mbox{sin}(2\pi s)s,\:\mbox{if}\:6\leq s<10.
\end{array}
\right.
\end{align}
A plot of the true spatial function $w_0(s)$ is provided in Figure~\ref{plots}. The function is piecewise differentiable which makes the estimation challenging.

We fit multiscale DCT with $J(1)=30$ to this dataset. As a competitor to multiscale DCT we implement\\
\textbf{DCT-GDP:} DCT-GDP uses the same basis functions as multiscale DCT, but replaces \emph{multiscale tree shrinkage prior} by
Generalized Double Pareto (GDP, \cite{armagan2013generalized}) shrinkage prior on the basis coefficients. GDP prior does not allow any multiscale structure, thereby asserting equivalent apriori shrinkage on all the basis coefficients.\\
\textbf{DCT-Normal:} DCT-Normal also uses the same basis functions with the prior on basis coefficients given by the independent normal prior distributions.

The two competitors are mainly aimed at comparing the inferential advantage of the tree shrinkage prior over the ordinary shrinkage prior and normal prior distributions. Additionally, we fit multi-scale DCT with one and two resolutions to assess how the choice of $R=3$ improves inference. Multiscale DCT with one and two resolutions are referred to as MDCT(1) and MDCT(2) respectively.
\begin{figure}[th!]
\begin{center}
\subfigure[Different resolutions]{\includegraphics[width=6cm]{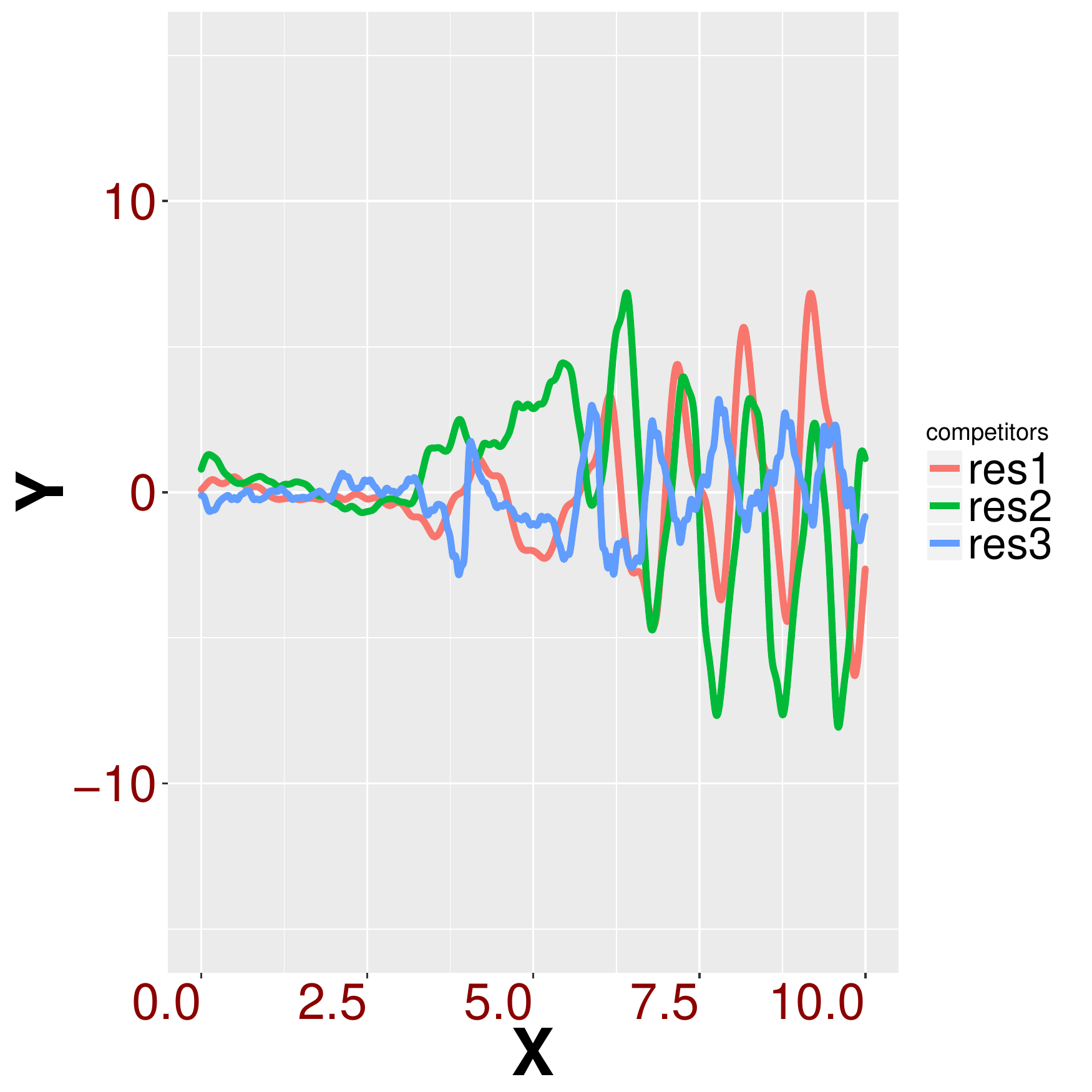}}
\subfigure[True vs. estimated surface]{\includegraphics[width=6cm]{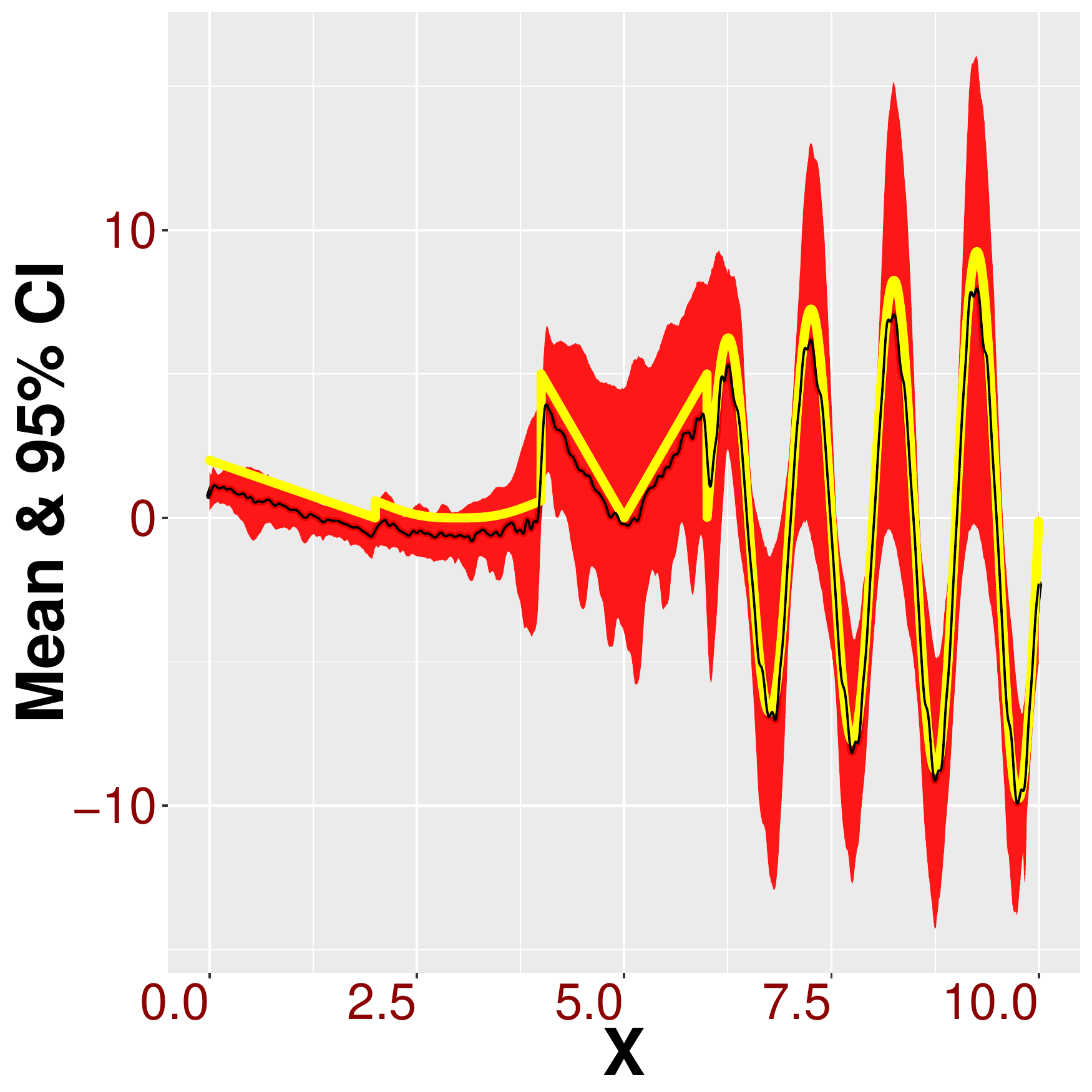}}
\subfigure[MSE]{\includegraphics[width=15cm]{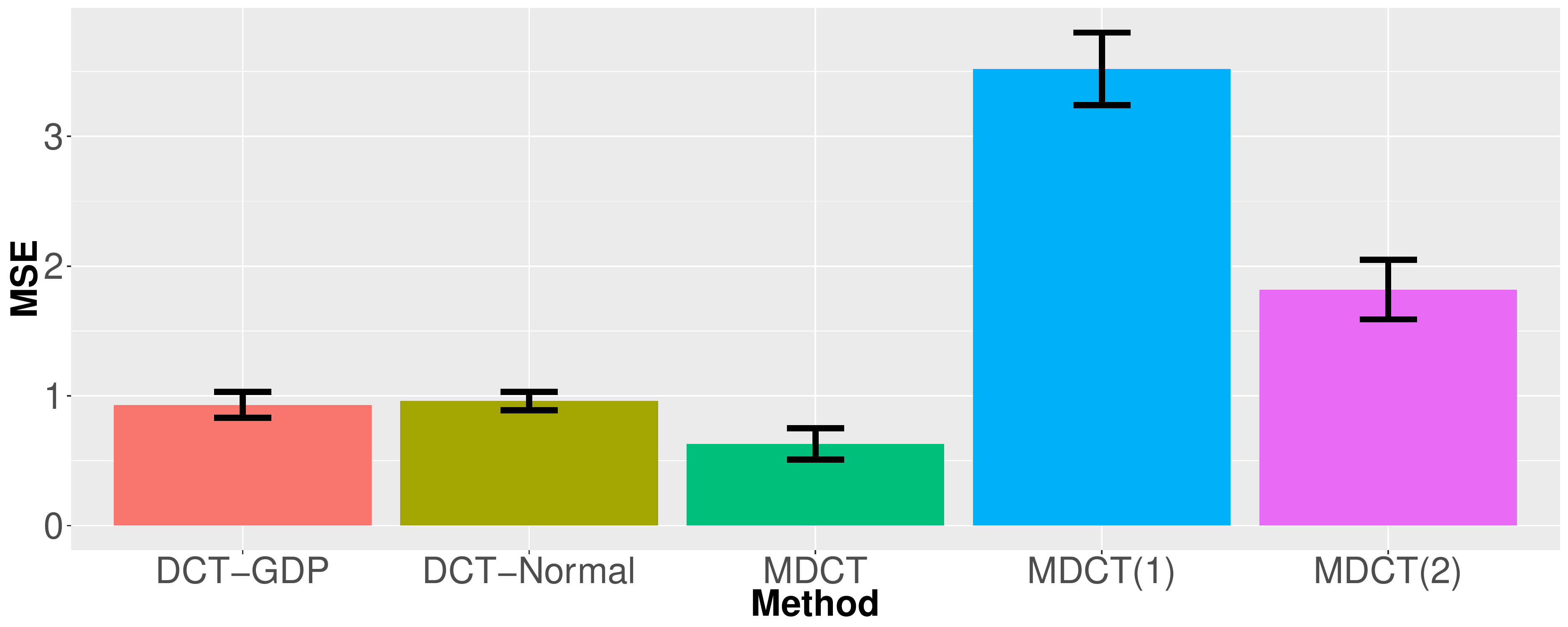}}
\end{center}
\caption{(a) Estimated mean function at different resolutions; (b) shows the true vs. the
estimated function in $R=3$ resolutions. The true function is in yellow and the estimated
function is in black. 95\% confidence bands for the estimated function are displayed in
red. (c) shows the MSE with associated standard errors for all competitors.}\label{plots}
\end{figure}

Figure~\ref{plots} reveals the role played by the three resolutions in estimating $w_0(s)$.
Clearly, resolutions 1 and 2 mostly capture global trends. While resolution 1 mostly captures
positive side of the sinusoidal curve, negative extremities of the sinusoidal curve is mostly
reconstructed by resolution 2. Resolutions 3 captures the
local variability found in the interval [4,10].

The inferential performance of MDCT is evaluated in estimating the spatial surface using
mean squared error. To be more precise, let $\hat{w}(s_i)$ be the posterior
median of $w(s_i)$. Define mean squared error (MSE) by
$MSE=\frac{1}{n}\sum_{i=1}^n(\hat{w}(s_i)-w_0(s_i))^2$. Average MSE along with the
associated standard errors over multiple simulations for the three competing models are
presented in Figure~\ref{plots}. It is evident from these figures that MDCT,
with the same number of knots and same basis functions, provide improved inference, the
reason being implementation of a structured prior distribution on the basis coefficients.
The computation times to implement the three competitors are about the same, with
one MCMC iteration in MDCT taking $\approx 0.33$  seconds to run the full scale inference. Additionally, there
seems to be a substantial improvement in terms of MSE with increasing resolutions, though
it stabilizes after $R = 3$.

The one dimensional exploration of MDCTs presented above, points at a few important advantages.
First of all, multi-scaling is able to capture local features succinctly, yielding
superior inference with similar number of knots and same basis functions over one scale
DCT. Secondly, the computational advantage of multiscale DCT is enormous given that full
Bayesian inference can be performed with a series of local computations. Given the
architecture of MDCT, it is possible to implement MDCT by storing
subsets of data in different processors. Moreover, it does not require storing large
dimensional covariance matrices that incurs onerous storage burden, as is the case for
Gaussian process based spatial models. In the next section a more involved comparative
analysis of MDCT with popular competitors is presented in the context of two dimensional
spatial examples.

\subsection{Two dimensional example}
\noindent In this section, we use two dimensional synthetic datasets to assess
the performance of MDCT in comparison to popular models for large spatial data. For the sake
of our exposition, MDCT is implemented with $3$ resolutions having a total of $2100$ basis
functions. As competitors to MDCT we implement: \\
(1) \textbf{Modified predictive process (MPP)}: a popular low rank
model fitted to the entire data with full Bayesian implementation \citep{finley2009hierarchical}; \\
(2) \textbf{LatticeKrig}: LatticeKrig
\citep{nychka2015multiresolution} is a recently proposed multiresolution model for big
data that employs kernel convolution with radial basis functions and Gaussian Markov
Random Field (GMRF) distribution on the basis coefficients. \texttt{LatticeKrig} package in \texttt{R} is employed for non-Bayesian implementation of LatticeKrig with $3$ resolutions having a
total $12678$ basis functions. LatticeKrig is a closely related competitor to MDCT with the major difference stems from using a GMRF prior on basis coefficients. \\
(3) \textbf{LaGP:} Local approximate Gaussian process \citep{gramacy2015local}.
LaGP has emerged from the computer experiment literature and is devised to
perform fast local neighborhood kriging with Gaussian processes.
LaGP is not designed to provide full scale Bayesian inference and is
only employed to compare predictive inference with other competitors.

We acknowledge \textbf{NNGP} \citep{datta2015hierarchical} and \textbf{multiresolution predictive process} \citep{katzfuss2016multi}
as the two important competitors of MDCT. However, both these methods are complicated in terms and implementation and till date
there is no readily available open source software to implement these methods. To avoid the risk of implementing them incorrectly,
we refrain from showing inferential results from these two methods. Moreover, it is argued in \cite{gramacy2015local} that LaGP performs
better than nearest neighbor methods in many applications.

Bayesian implementation of MPP is performed using the package \texttt{spBayes} in
\texttt{R}. It is well known that MPP is not a suitable model when the sample size
is very large. Therefore to accommodate MPP as a competitor in the analysis, we design
simulation studies for datasets not larger than $\sim 10500$ locations. Note that
as the number of knots increases, performance of MPP improves, though adding much burden
to computation. Therefore, while comparing MPP, we focus both on computation time and
accuracy. Finally, the \texttt{laGP} package facilitates frequentist implementation of LaGP in
\texttt{R}. All the interpolated spatial surfaces are obtained using the \texttt{R}
package \texttt{MBA}.

To illustrate the performance of the competitors, $10,500$ observations within $[0,1]\times[0,1]$ domain are generated from the classical geostatistical model with likelihood $\by\sim N(\bX\bgamma+\bw_0,\sigma^2\bI)$. The model includes an intercept $\gamma_0$ and a predictor $\bx(s)$ drawn i.i.d from from $N(0,1)$ with the corresponding coefficient $\gamma_1$, $\bgamma=(\gamma_0,\gamma_1)$. $\bw_0=(w_0(\bs_1),...,w_0(\bs_n))'$ is an $n$ dimensional vector that follows a multivariate normal distribution with mean $\bzero_n$ and the covariance matrix of the order $n\times n$ with $(i,j)$ entry given by $\upsilon(\bs_i,\bs_j,\theta_1,\theta_2,\nu))$. $\upsilon$ is chosen from the popular Matern class of correlation functions given by
\begin{align}\label{matern}
\upsilon(s,s',\theta_1,\theta_2,\nu)=\frac{\theta_1}{2^{\nu-1}\Gamma(\nu)}(||s-s'||\theta_2)^{\nu}\mathcal{K}_{\nu}(||s-s'||\theta_2);\:\theta_2>0,\:\nu>0,
\end{align}
with $\theta_2,\nu$ controlling spatial decay and process smoothness respectively, $\Gamma$ is the usual Gamma function and $\mathcal{K}_{\nu}$ is a modified Bessel function of the second kind with order $\nu$ \citep{stein2012interpolation}. We fixed $\nu=0.5$ which reduces to the
exponential covariance kernel and generates continuous but non-differentiable sample paths. For simulations, we use $\theta_2=3$ and the ratio of spatial to noise variability is kept at $20$. Among $10,500$ observations, $10000$ are randomly selected for model fitting and the rest are kept as a test dataset to facilitate predictive inference.

\begin{figure}[th!]
\begin{center}
\subfigure[True surface]{\includegraphics[width=5cm]{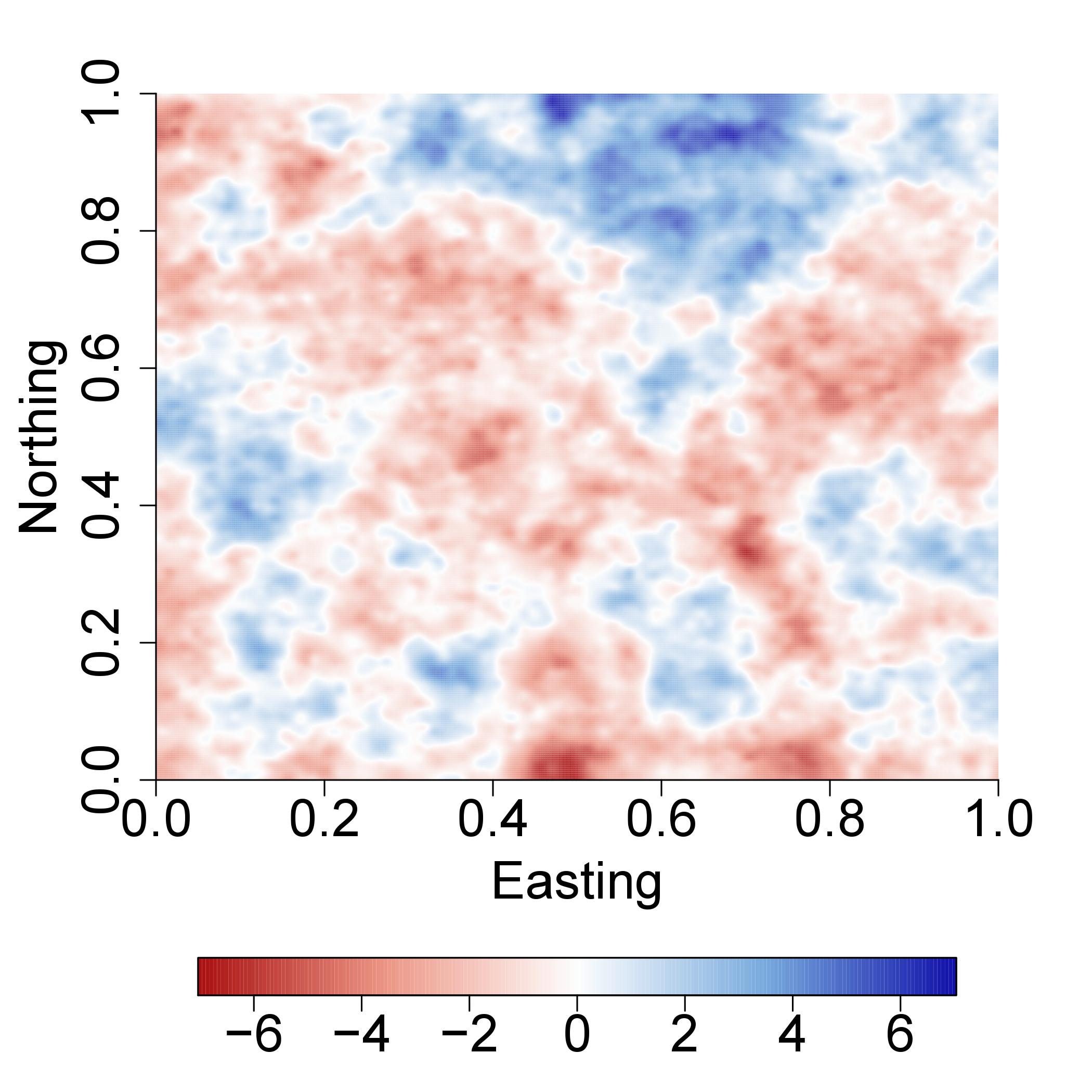}}
\subfigure[LatticeKrig]{\includegraphics[width=5cm]{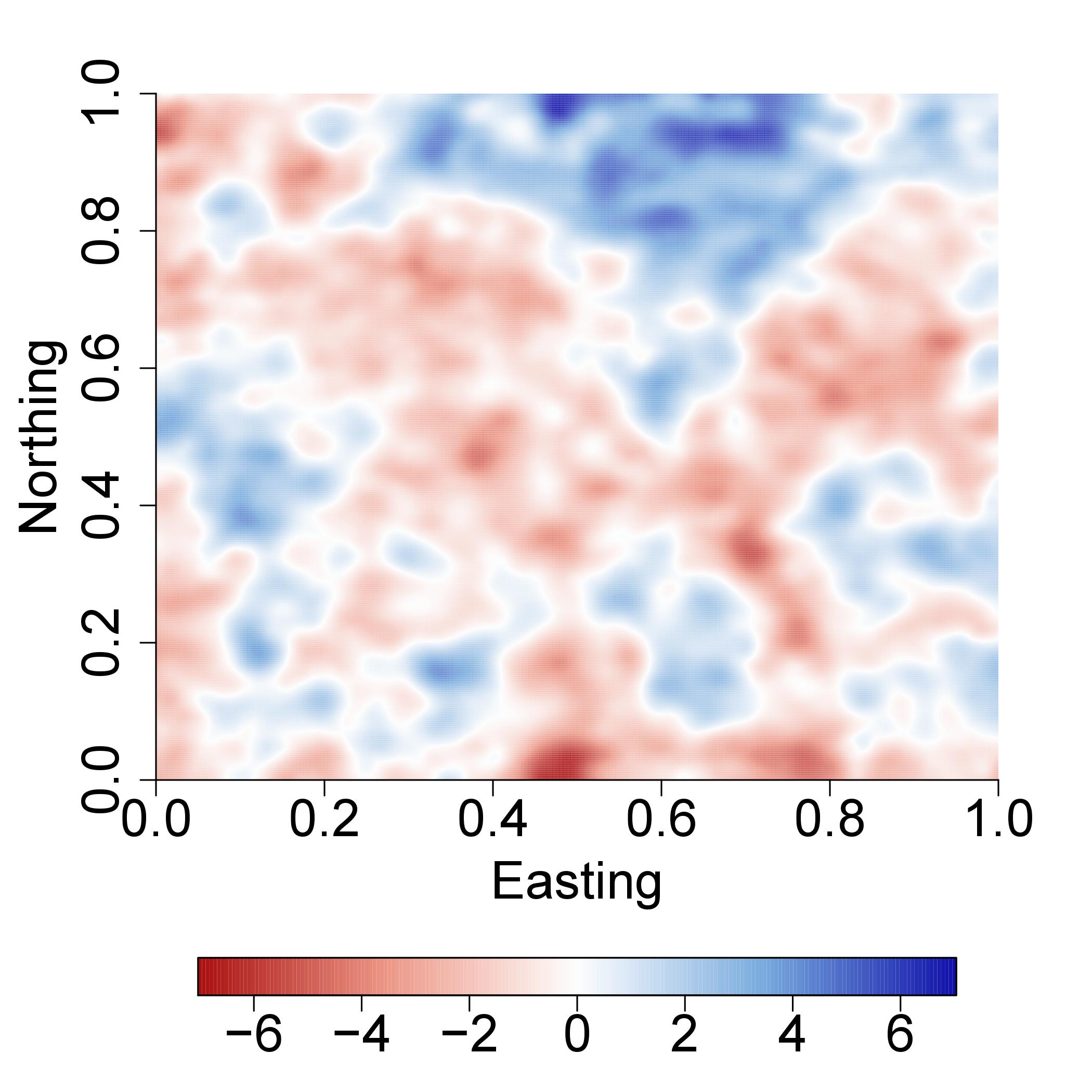}}
\subfigure[MPP]{\includegraphics[width=5cm]{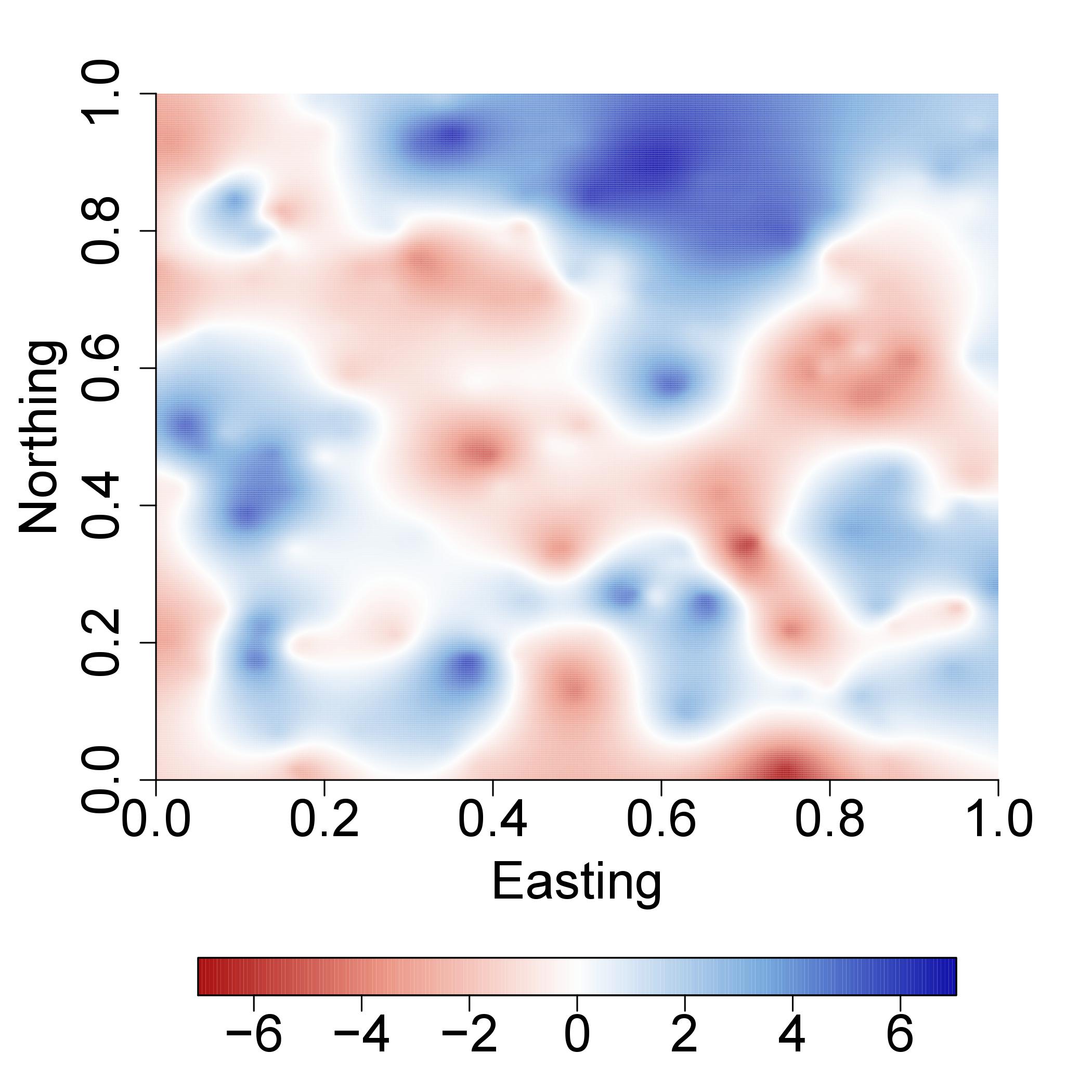}}\\
\subfigure[MDCT mean surface]{\includegraphics[width=5cm]{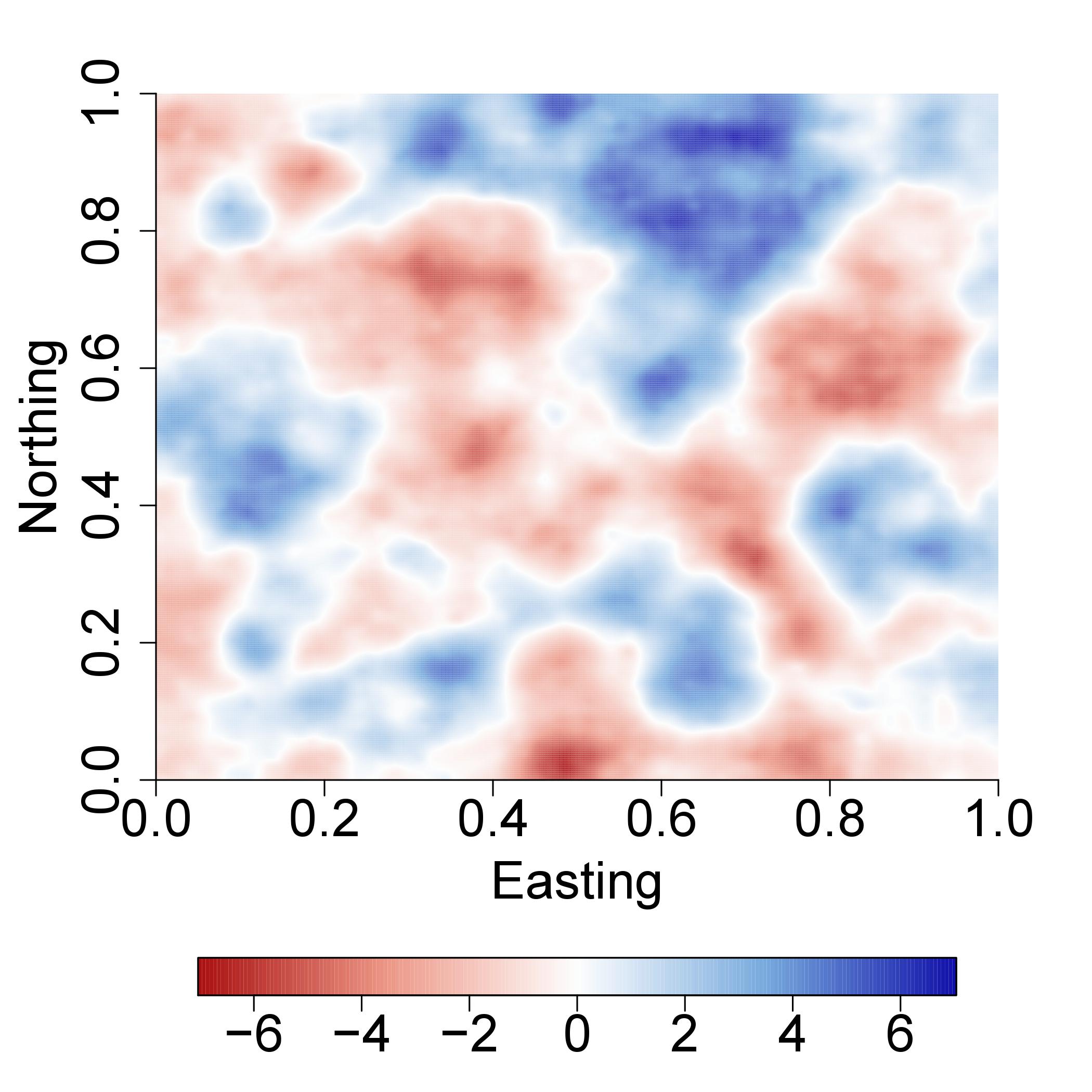}}
\subfigure[MDCT upper 95\% CI]{\includegraphics[width=5cm]{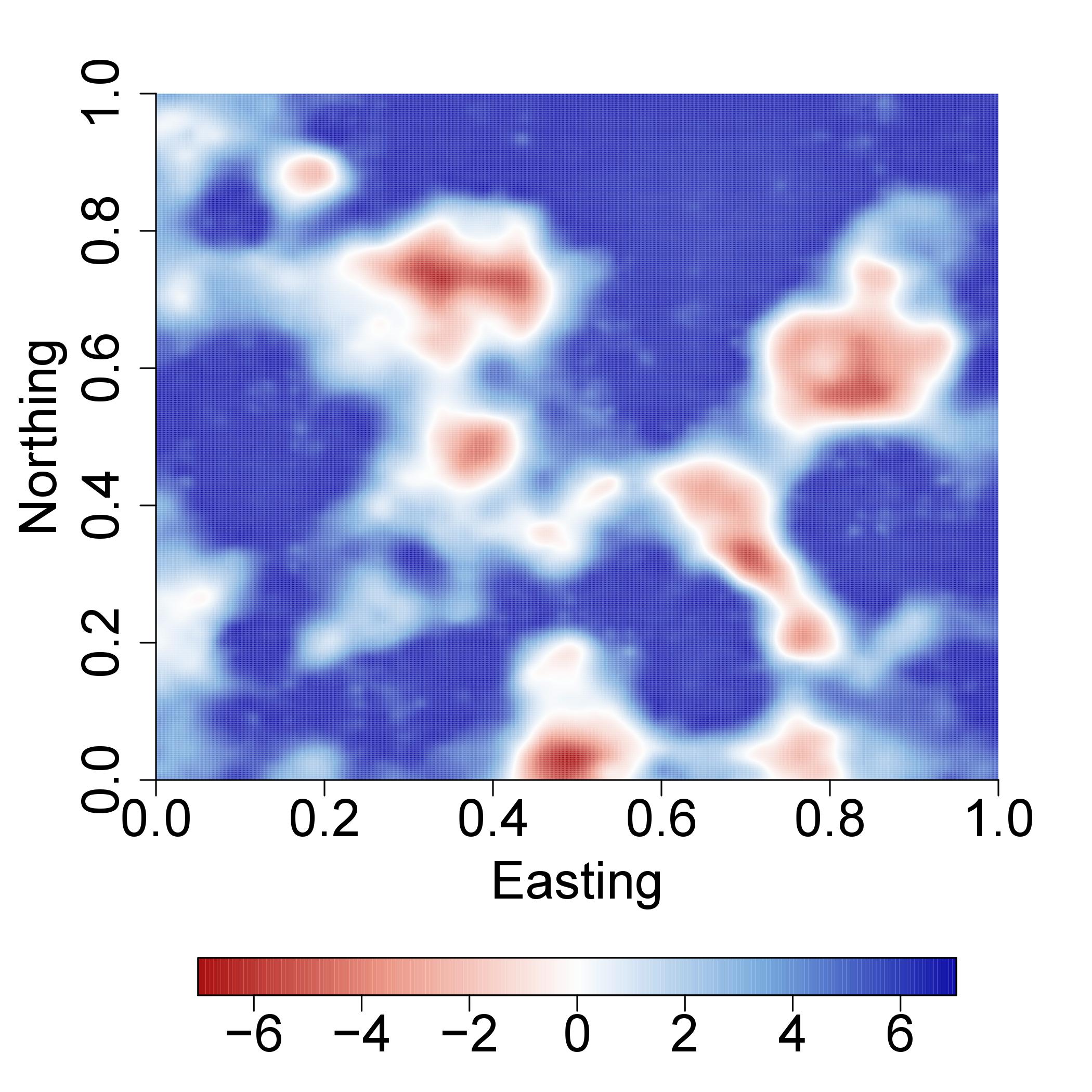}}
\subfigure[MDCT lower 95\% CI]{\includegraphics[width=5cm]{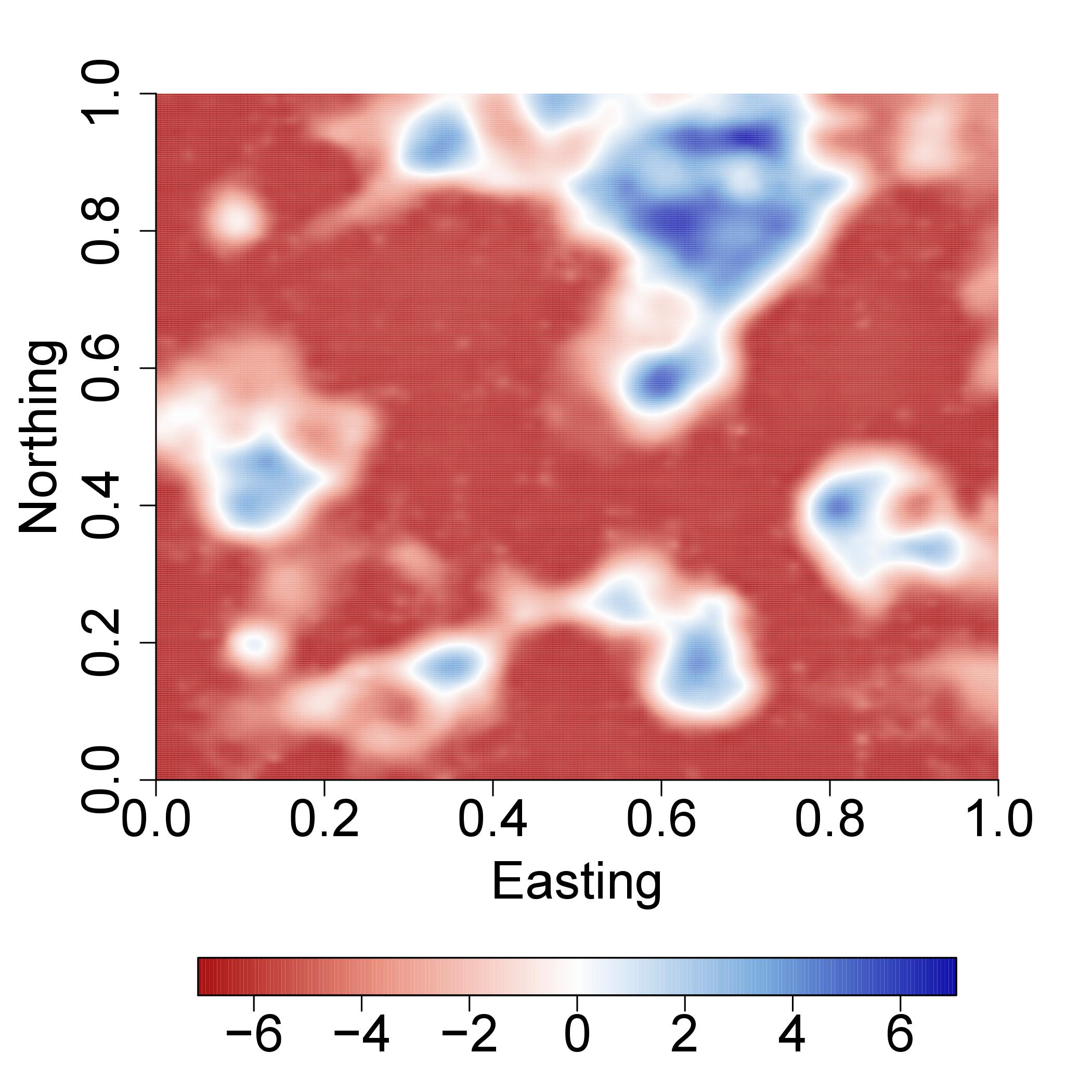}}\\
\subfigure[MSE]{\includegraphics[width=10cm]{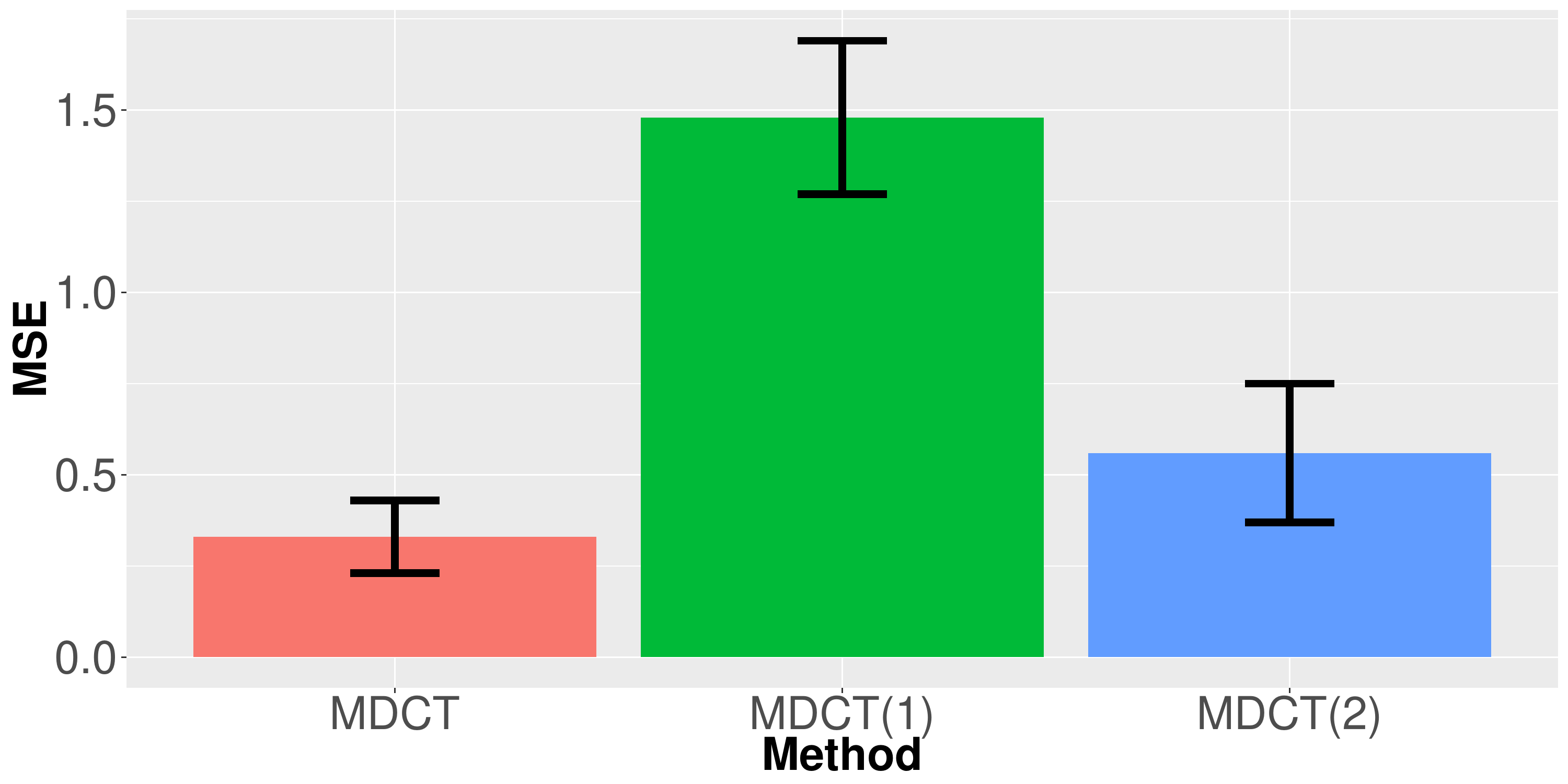}}
\end{center}
\caption{(a) True data generation surface along with posterior mean residual surface from (b) LatticeKrig (c) modified predictive process (d) multiscale DCT and; (e) and (f) present estimated 95\% upper and lower quantile surfaces; (g) MSE along with standard errors for MDCT, MDCT(1) and MDCT(2)}\label{plots2d}
\end{figure}
Figure~\ref{plots2d} presents true data generating surface and the estimated residual surfaces for LatticeKrig (LK), MDCT and MPP. MPP shows excessive smoothing, while LatticeKrig and multiscale DCT yield essentially equivalent estimates of the spatial surface. Further MSE for multiscale DCT (MDCT) with one resolution (MDCT(1)) and two resolutions (MDCT(2)) are also plotted in Figure~\ref{plots2d}.
As expected, MSE for MDCT decreases as the number of resolutions increases.

Next, we turn our attention to the predictive inference of the competitors. To this end, we compare all competitors based on their ability to
produce accurate point prediction and predictive uncertainties. Point prediction of the
competitors are judged based on the mean squared prediction error (MSPE) metric. For
Bayesian competitors, predictive uncertainties are characterized by the length and
coverage of 95\% predictive intervals. The frequentist implementation of LaGP and LK
provides only predictive point estimates and standard errors. Thus, for these two competitors, approximate 95\% predictive intervals are constructed by considering [predictive point estimate $-1.96 *$ standard error, predictive point estimate $+1.96 *$ standard error].

It is clear from Figure~\ref{Fig:MSPE_TIME} that MDCT yields a MSPE that is competitive
with those of LK and LaGP, though the latter two slightly outperforms MDCT.
Interestingly, MDCT with $R=3$ resolution shows significantly improved performance
compared to MDCT with resolution one or two. Intuitively, one can explain such an upsurge
in performance by noting that the true surface generated from the exponential correlation
function has significant local behavior, thereby limiting the ``borrowing of information"
across space. Moreover, coverage and length of 95\% predictive intervals of MDCT demonstrates
accurate characterization of predictive uncertainly as opposed to MPP which shows some
under-coverage. Likewise, the approximate 95\% predictive interval of LaGP exhibits little
under-coverage, while for LK we observe massive under-coverage.

To check the sensitivity with respect to the choice of $R=3$, we run our analysis with
$R=4,5$ and compare to the MSPE obtained from $R=3$. Table~\ref{predictionR} clearly shows that
beyond $R=3$ the improvement in MSPE performance is not commensurate with the increase in
computation cost. We found this conclusion to hold across a number of simulation studies.
Therefore $R=3$ is kept throughout this article.
\begin{figure}[!ht]
  \begin{center}
    \includegraphics[width=8cm]{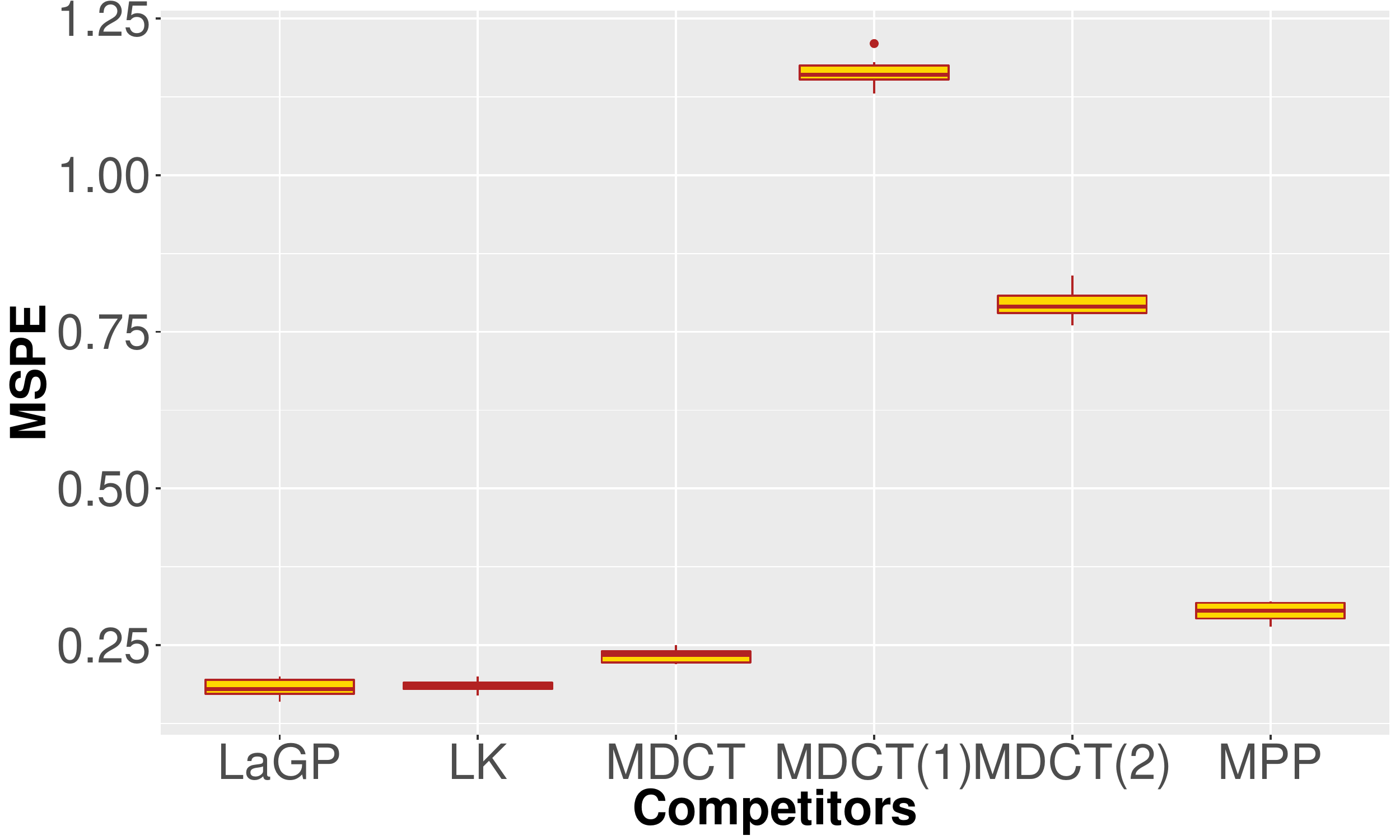}\\
    \includegraphics[width=8cm]{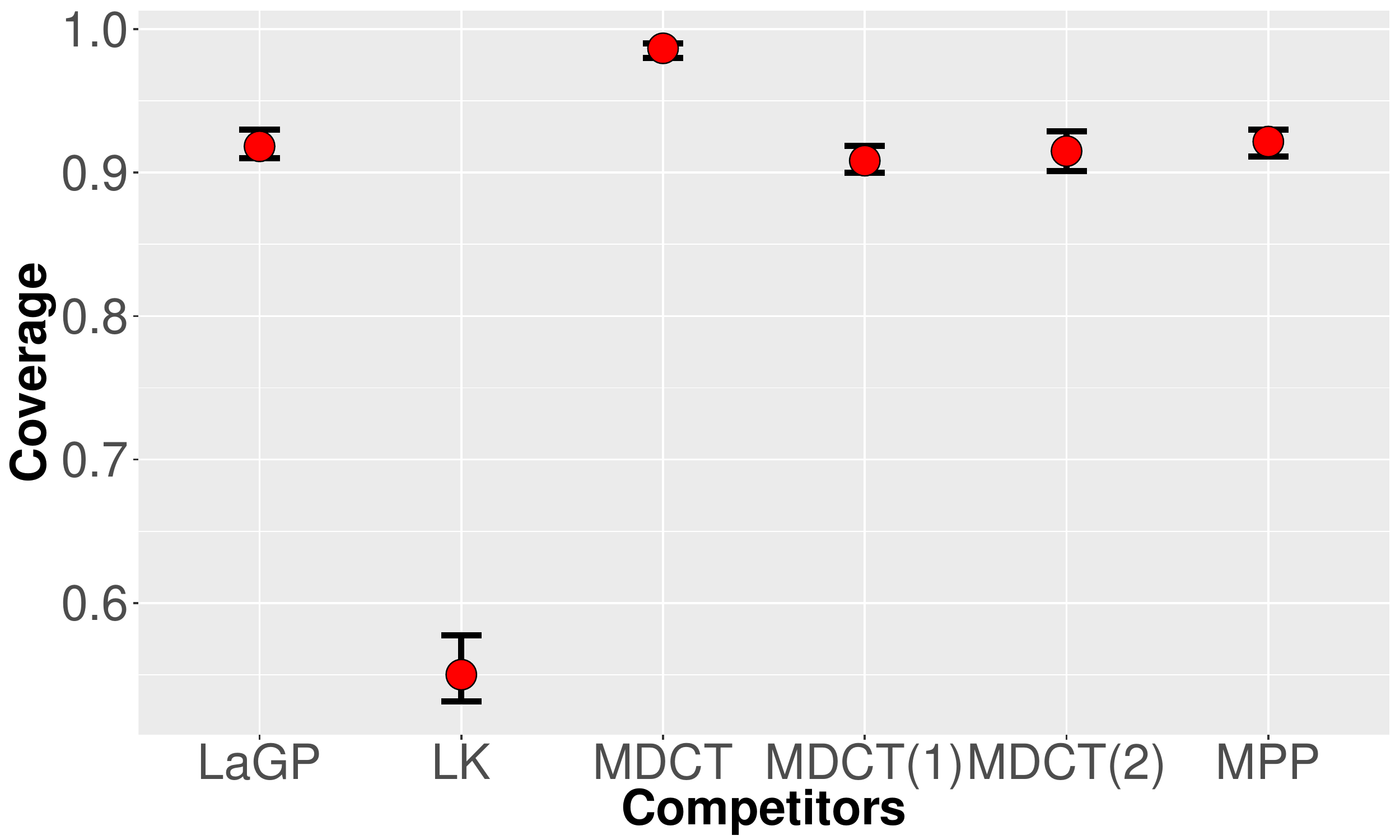}
    \includegraphics[width=8cm]{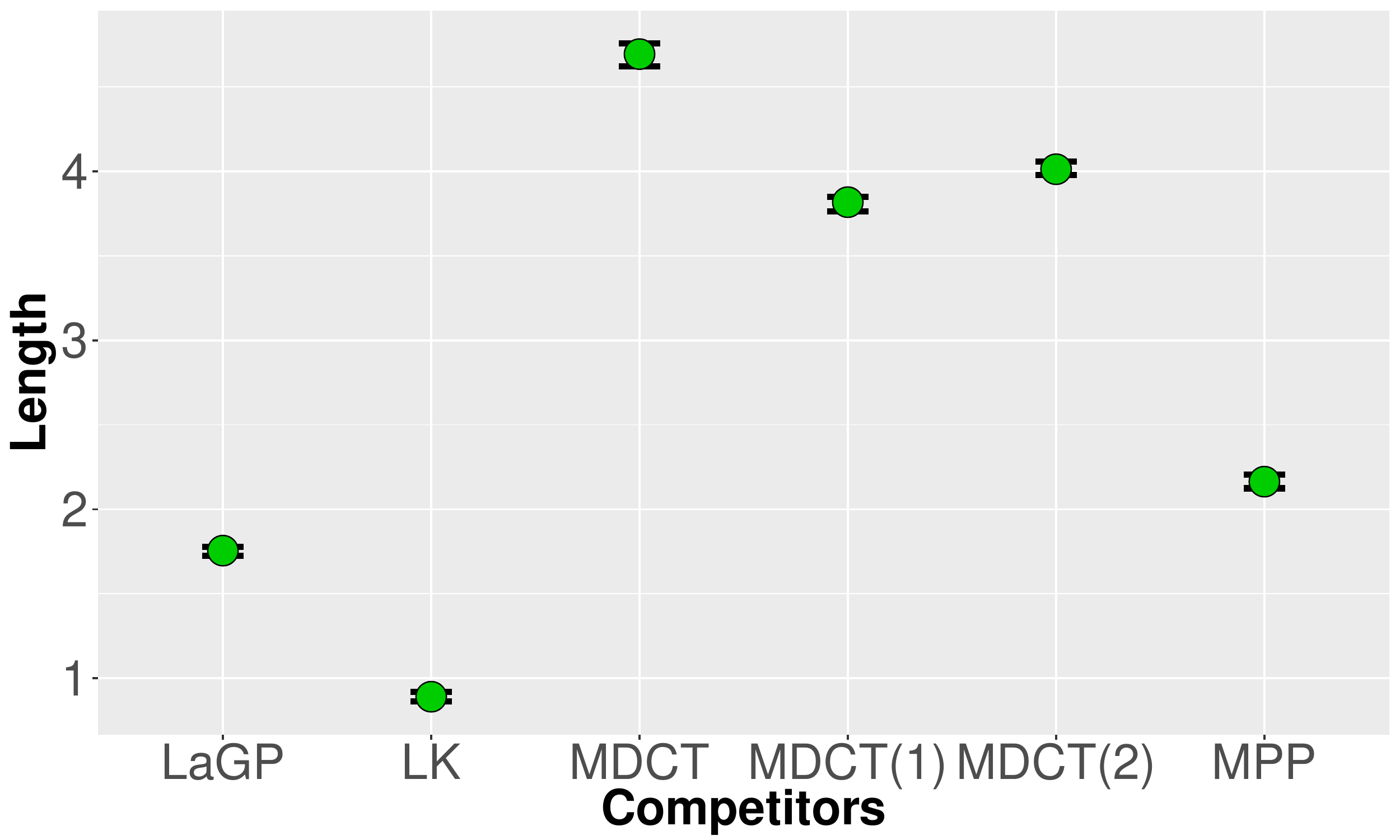}
  \end{center}
\caption{Plot at the top indicates boxplot of mean squared prediction error for all competitors over a few replications. Second and third plots show coverage and length of 95\% predictive intervals for the competitors over the same replications. LatticeKrig shows extreme under-coverage compared to others.}\label{Fig:MSPE_TIME}
\end{figure}

\begin{table}[h]
{\footnotesize
\centering
\begin{tabular}{cccc}
	\hline
 $MDCT$ & $R=3$ & $R=4$ & $R=5$\\
	\hline
MSPE & $0.25_{0.03}$ & $0.21_{0.02}$  & $0.19_{0.02}$ \\
	\hline
\end{tabular}
\caption{Average mean squared prediction error for MDCT with $R=3,4,5$. Associated standard errors over 5 repeated simulations are provided in the subscript.}\label{predictionR}
}
\end{table}

Our analysis finds some interesting points about the competing methods. Recall that MDCT and LatticeKrig are similar in terms of their
multiscale structure, the only difference being the distribution on the basis coefficients. Our investigation reveals that tree shrinkage
priors are appropriately calibrated so as to yield similar point estimates with LatticeKrig with much less number of basis functions. The
GMRF prior distribution on basis coefficients hinders efficient local computation in LatticeKrig. As a result
LatticeKrig in its present form has less scope of being computationally efficient. In contrast, multiscale DCT is able draw full scale Bayesian
inference with a series of parallelizable local computations. At the same time, tree shrinkage prior in multiscale DCT penalizes model overfitting with unnecessary resolutions. Similarly, LaGP has a notable disadvantage for not being model based. As a result, it is not clear
how to extend LaGP for non-Gaussian data, while MDCT structure can readily be embedded into a hierarchical structure to model non-Gaussian spatial data, as is described in the next section.

\noindent \underline{\emph{Computation Time:}}
MDCT in this specific example takes approximately $3.07$ seconds per iteration
with non-optimized, non-parallel \texttt{R} implementation, while MPP implemented in
\texttt{C++} takes close to $7.2$ seconds to run one MCMC iteration. We notice, though, that MPP performs the estimation of the basis functions,
while MDCT assumes a fixed form with the empirical Bayes estimate of $\eta$ at every iteration.
However, even with fixed basis functions, MDCT is able
to demonstrate superior inference to MPP. It is possible that, performing more elaborate inference on
the kernel parameters of the MDCT would improve the predictive performance of the model.
This would come at the cost of increased computational complexity, thus the benefits of
such extension is unwarranted. Recall that, in this specific example MDCT is implemented
with $J(1)=400$. To understand how the computation time of MDCT varies vis a vis MPP with
changing $n$ and $J(1)$, we implement both MPP and MDCT with $J(1)=5^2,10^2$ for
different sample sizes.
 Figure~\ref{comp} reports the computation time for the
competitors using the \texttt{R} function \texttt{Sys.time}. It is be noted that MDCT can be implemented either by
sequentially updating $J(1)$ blocks of parameters or by parallelly updating these $J(1)$ blocks independently in $J(1)$ nodes. Thus the figure displays computation time with both parallel and non-parallel implementation of the MDCT model. The blue
line shows the computation time for MPP from $n=1000$ to $n=100,000$, while the green and
red lines show the same for MDCT with $J(1)=5^2$ and $J(1)=10^2$ respectively. Clearly,
the computation time for MDCT increases linearly with $n$ for both cases. It is also
noticed that the computation time of MDCT with $J(1)=5^2$ is about 4-5 times faster than
$J(1)=10^2$.  The increase in computation time is due to sequential updating of
parameters in $J(1)$ blocks. With a proper parallelized implementation of MDCT, arguably
the increase in computation time from $J(1)=25$ to $J(1)=100$ will be minimal. We found that
practical implementation of MPP becomes prohibitive due to both, memory allocation and
exorbitant computation time, for $n$ above $100,000$. On the contrary, MDCT facilitates
distributed storage of big data into multiple processors. It is worth noticing that
frequentist implementation of LK and LaGP draw inference for a point estimate within a
few minutes. In summary, 2D simulation examples comprehensively establishes MDCT as an
effective tool for fast Bayesian implementation of large scale spatial data.

\begin{figure}[th!]
\begin{center}
\includegraphics[width=14cm,height=6cm]{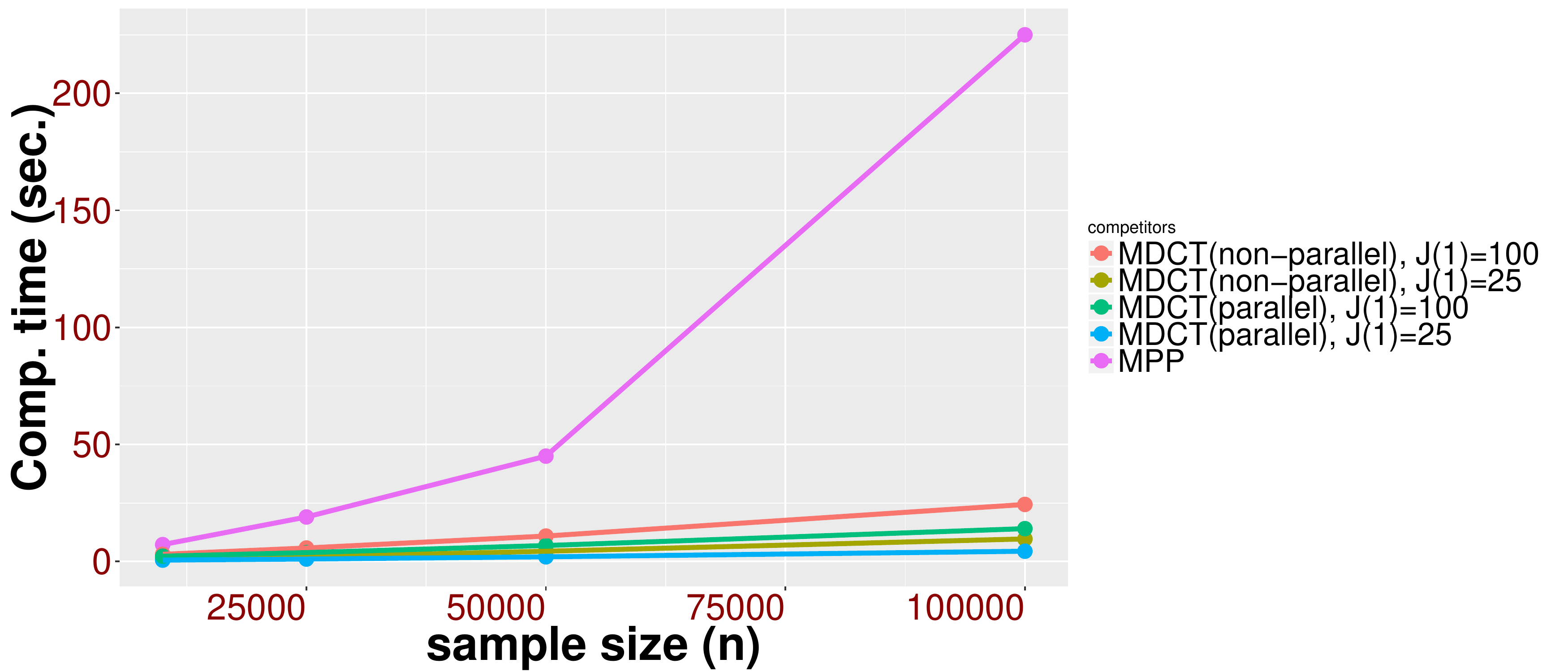}
\end{center}
\caption{Computation time for MPP with $200$ knots, MDCT with $J(1)=25$ and $J(1)=100$. Computation times per MCMC iteration are presented for both MDCT and MPP.}\label{comp}
\end{figure}

\subsubsection{Two dimensional illustration of MDCT with binary spatial data}\label{binary}
To demonstrate the flexibility offered by MDCT as opposed to ad-hoc predictive methods (such as
LaGP), performance of MDCT is investigated under non-Gaussian binary spatial data. For this purpose
$10,500$ observations within $[0,1]\times[0,1]$ domain are generated from the probit spatial regression
model. More precisely, with $\bx(\bs_i)$ as the predictor vector at $\bs_i$, the response $y_i$ is simulated
using
\begin{align*}
y_i &\stackrel{ind}{\sim} Ber(p_i)\\
\Phi^{-1}(p_i) &= \bx(\bs_i)'\bgamma+w_0(\bs_i).
\end{align*}
The model includes an intercept $\gamma_0$ and a predictor $\bx(\bs)$ drawn i.i.d from from $N(0,1)$ with the corresponding coefficient $\gamma_1$, $\bgamma=(\gamma_0,\gamma_1)$. $\bw_0=(w_0(\bs_1),...,w_0(\bs_n))'$ is an $n$ dimensional vector that follows a multivariate normal distribution with mean $\bzero_n$ and the covariance matrix of the order $n\times n$ specified through the Mat\'ern (\ref{matern}) class of correlation functions. A random subset of $10000$ observations are selected for model fitting and the rest is used to judge performance of
MDCT as a binary classifier.

For the sake of our exposition, MDCT is implemented with $3$ resolutions having a total of $2100$ basis
functions. Note that the binary regression precludes the possibility of employing LaGP as a competitor. On the other hand, \texttt{LatticeKrig} package implements LatticeKrig only for continuous response. Thus as a competitor, binary spatial regression with modified predictive process is implemented in \texttt{R} package \texttt{spBayes}.

\begin{figure}[!ht]
\begin{center}
    \subfigure[True surface]{\includegraphics[height=5cm,width=5cm]{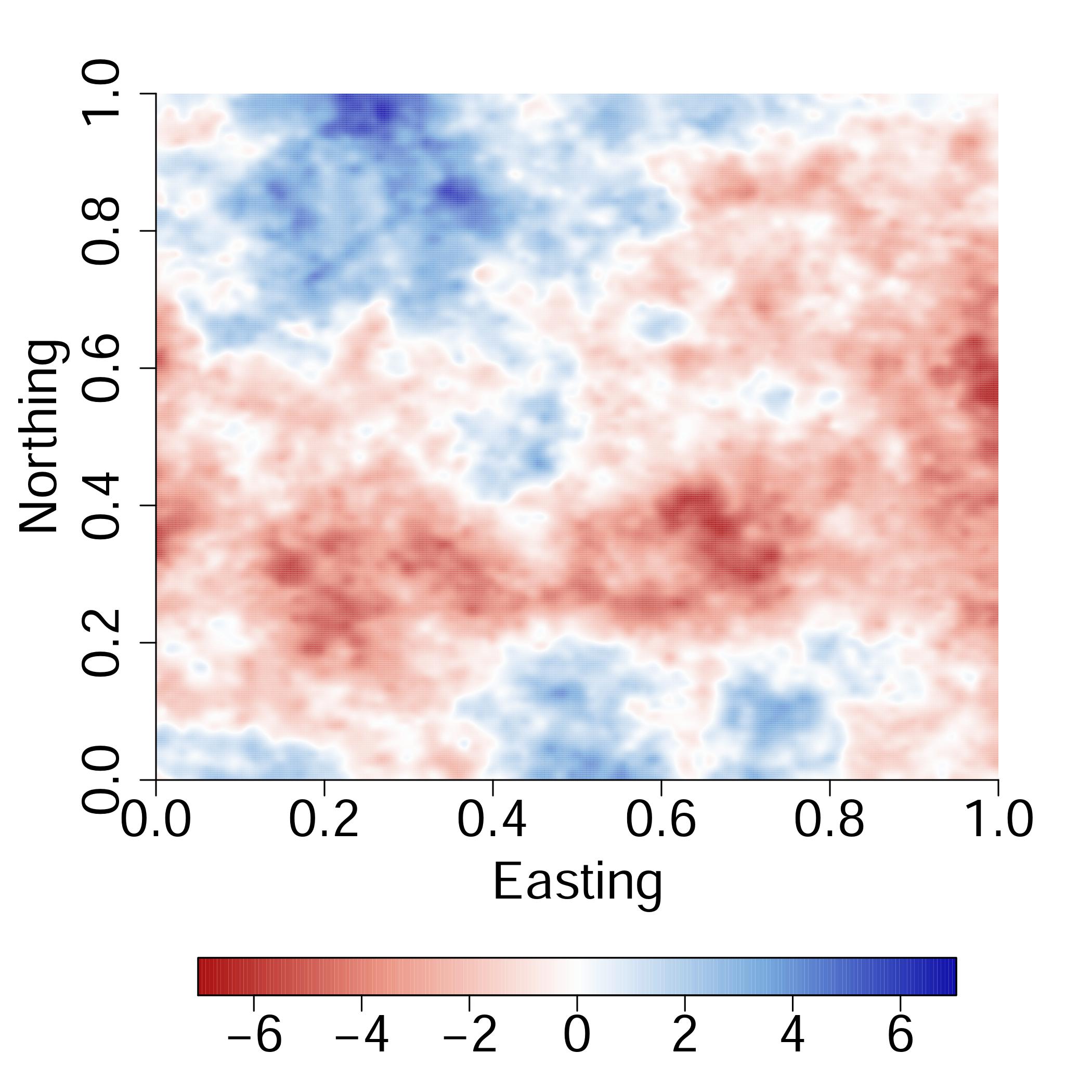}{\label{binary-figs:a1}}}
    \subfigure[Estimated surface: MDCT]{\includegraphics[height=5cm,width=5cm]{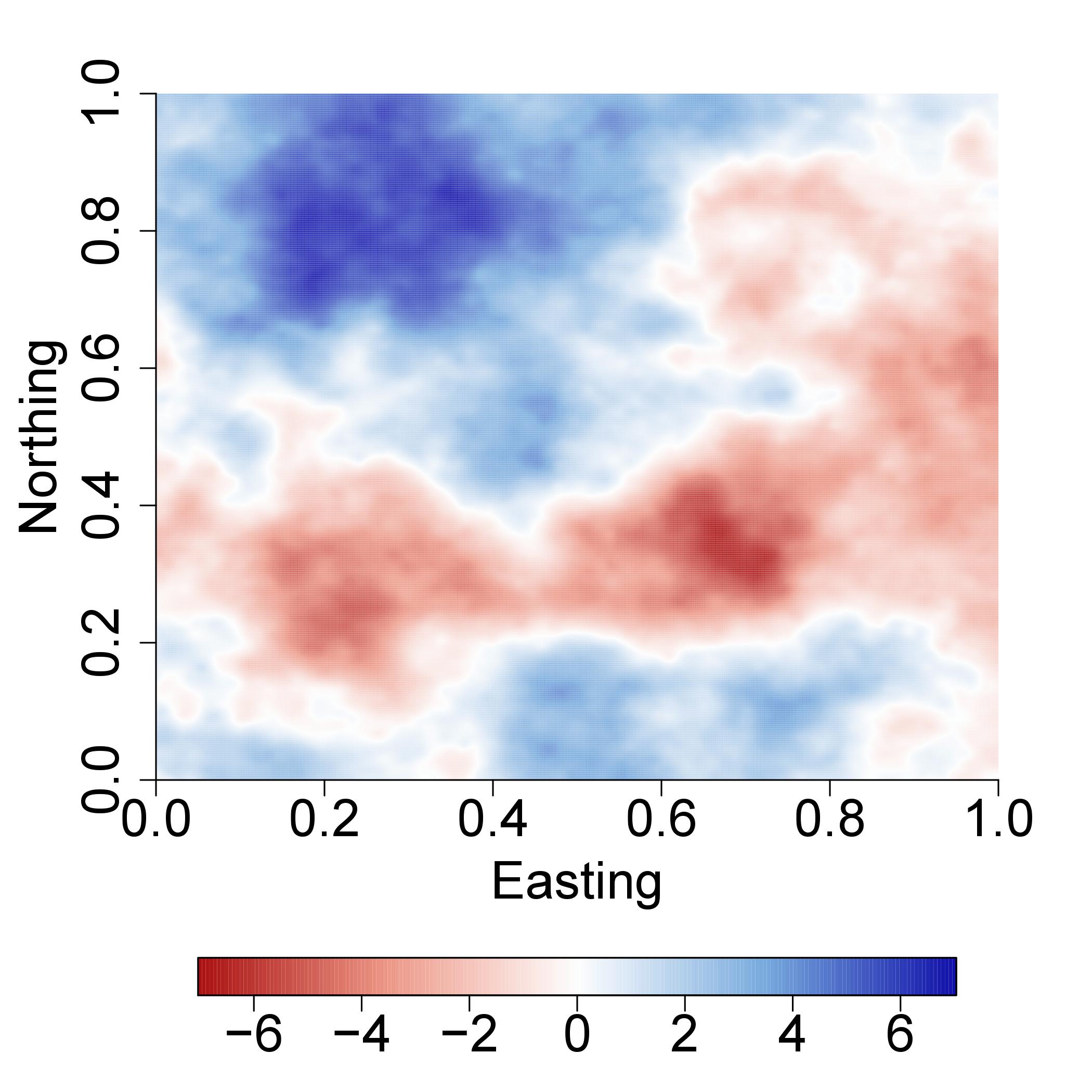}{\label{binary-figs:a2}}}
    \subfigure[Estimated surface: MPP]{\includegraphics[height=5cm,width=5cm]{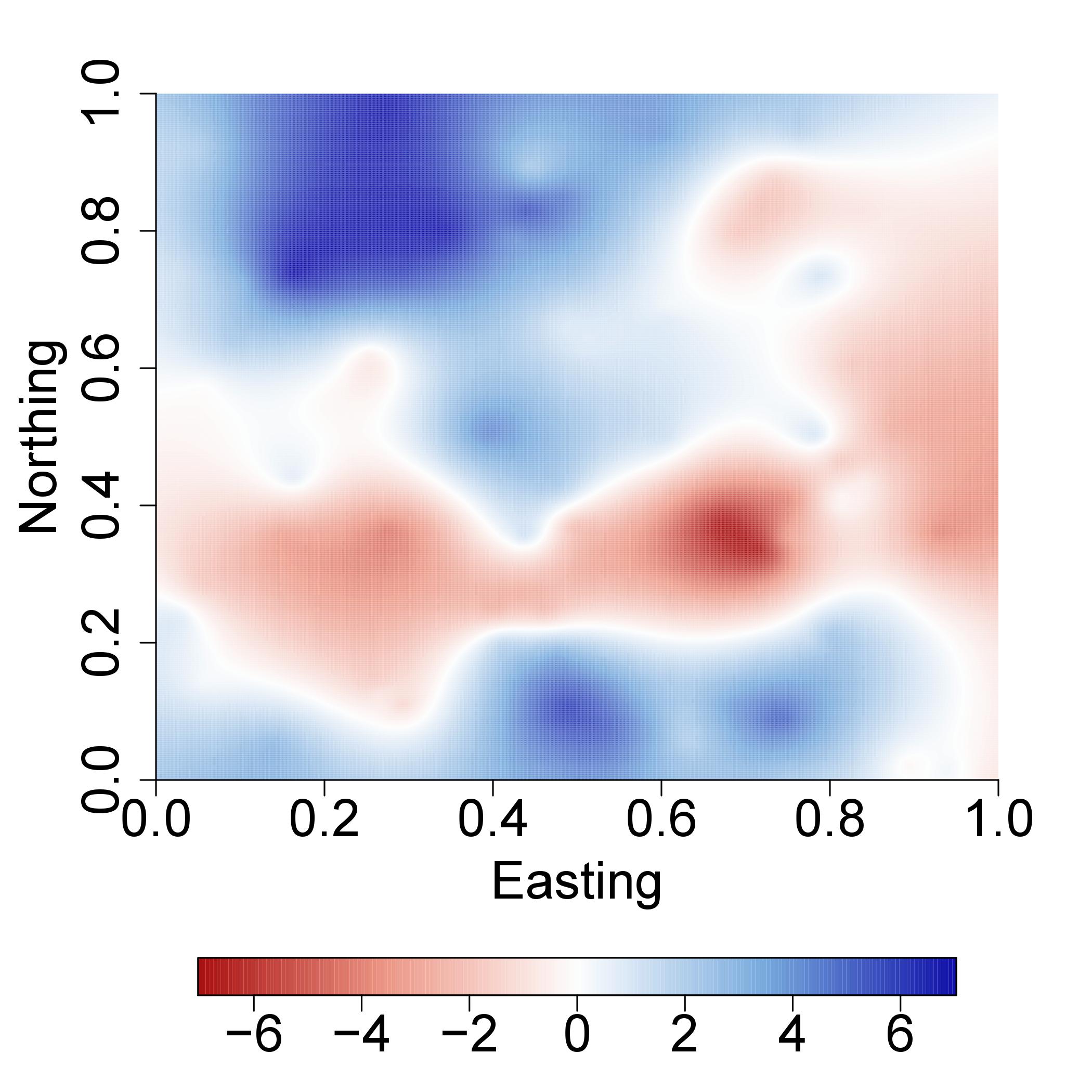}{\label{binary-figs:a4}}}\\
    \subfigure[ROC out of sample]{\includegraphics[height=5cm,width=5cm]{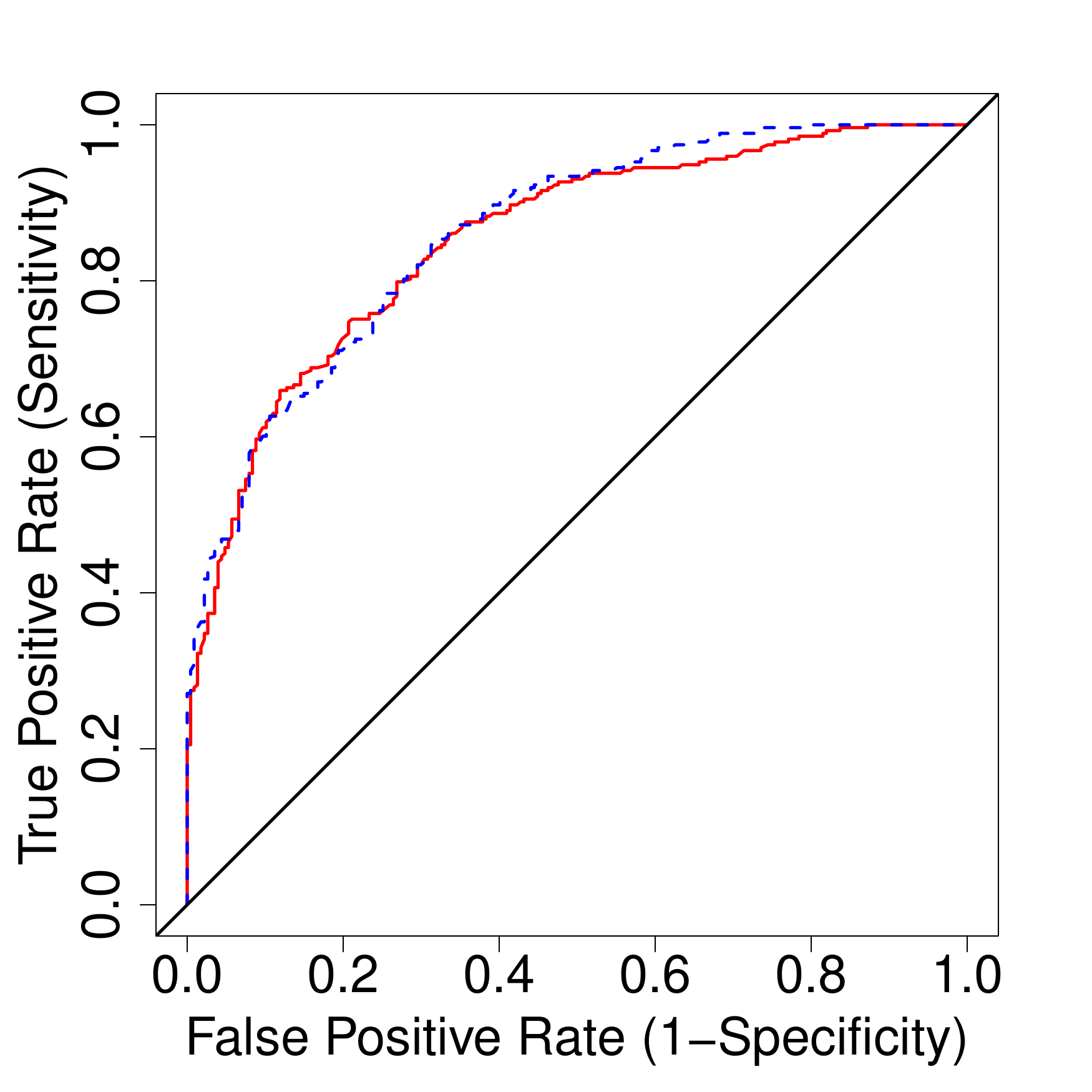}{\label{binary-figs:a3}}}
  \end{center}
\caption{(a) True surface generating the data. Figures (b) and (c) present the posterior predictive mean of estimated spatial surfaces
from MDCT and MPP. (d) shows out of sample ROC curves for MPP and MDCT. Dotted line presents ROC for MDCT, while solied line presents ROC for MPP.}
\label{binary}
\end{figure}

Figure~\ref{binary} shows the true surface and estimated surfaces from MDCT and MPP. Since the surface estimation from binary spatial regression is a notoriously challenging problem, it comes with no surprise that the performance of all
competitors deteriorate when compared with Figure~\ref{plots2d}. However, among the two competitors MDCT outperforms MPP considerably. It becomes clear from Figure~\ref{binary} that MPP undergoes massive oversmoothing and loses most of the local features in the spatial surface. MDCT also experiences smoothing, though to a much lesser degree than MPP. Referring to Figure~\ref{binary-figs:a3}, MDCT appears to be marginally better than MPP in terms of out of sample classification (Area under the ROC curve for MDCT is $0.74$, while the same for MPP is $0.68$). The binary spatial regression analysis further corroborates
the flexibility and accuracy of MDCT. Next section discusses performance of MDCT along with its competitors on a sea surface temperature dataset.

\section{Analysis of the sea surface temperature data}\label{SST}

A description of the evolution and dynamics of the oceans' temperature
is a key component of the study of the Earth's climate. Historical
records of ocean data have been collected for the purpose of
understanding the properties of water masses and their changes in time.
They are also used to assess, initialize and constrain numerical models
of the climate. Sea surface temperature data from ocean samples have
been collected by voluntary observing ships, buoys, military and scientific
cruises for decades. During the last 20 years or so, this wealth of
data has been complemented by regular streams of remotely sensed
observations from satellite orbiting the earth. Increasingly
sophisticated climatological research requires, not only the
description of the mean state and the relevant trends in ocean data,
but also a careful quantification of the  data variability at
different spatial and temporal scales. A number of articles have
appeared to address this issue in recent years, see e.g.
\cite{higdon1998process}, \cite{lemos2009spatio},
\cite{lemos2006spatio}, \cite{berliner2000long}.

This article considers the problem of capturing the spatial trend and
characterizing the uncertainties in the sea surface
temperature (SST) in the West coast of mainland USA, Canada and Alaska between $30^0-60^0$ N. latitude and $122^0-152^0$ W. longitude.
The dataset is obtained from  NODC World Ocean Database 2016 and we use
the data collected in the month of October for all the spatial
locations. Note that, for this example, we ignore the temporal
component. We perform screening of the data to
ensure quality control and then choose a random subset of $113,412$
spatial observations over the domain of interest. Out of the total
observations, about 90\%, i.e $100,000$ observations are used for model
fitting and rest are used for prediction. We replicate this procedure
$5$ times to eliminate any chance factor in our analysis. The domain of
interest is large enough to allow considerable spatial variation in SST
from north to south and provides an important first step to extend
these models for the analysis of global scale SST database.

The plot of the sea surface temperature along with coastal lines of
Western United States and Canada is shown in Figure~\ref{data-figs:a1}.
The data show a clear decreasing SST trend with increasing
latitude. Consequently, we add latitude and longitude as linear
predictors to explain the long-range directional variability in the
SST. We fitted a non-spatial model with latitude and longitude as linear
predictors using ordinary least square (OLS). The resulting residuals
are shown in Figure~\ref{data-figs:a2}. The residual plot reveals
spatial dependence with no obvious pattern of aniosotropy. Thus a
multiscale DCT model with latutude and longitude as predictors seem to
be a desirable model for this data.

\begin{figure}[!th]
  \begin{center}
    \subfigure[Sea surface temperature]{\includegraphics[height=5cm,width=5cm]{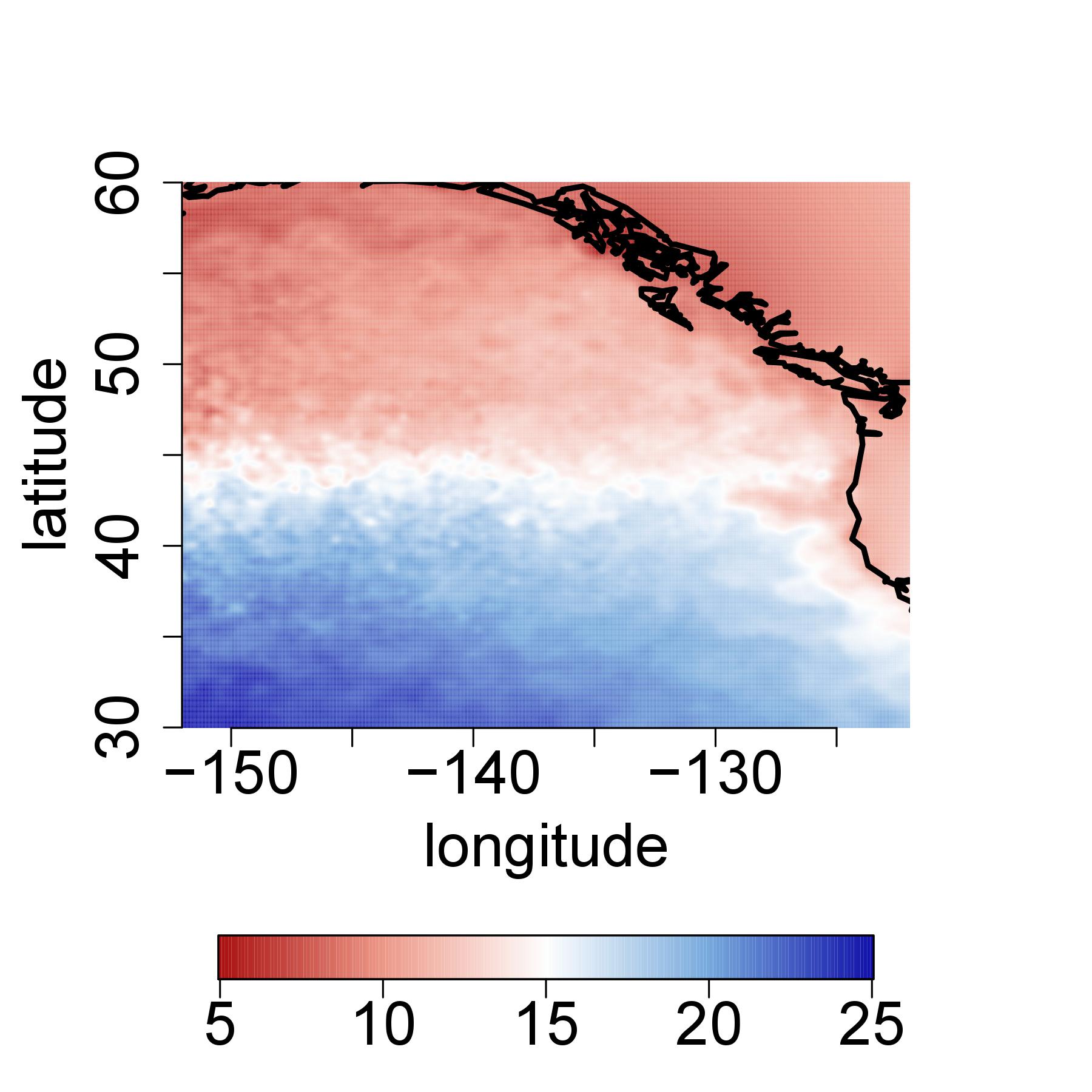}{\label{data-figs:a1}}}\\
    \subfigure[OLS residual]{\includegraphics[height=5cm,width=5cm]{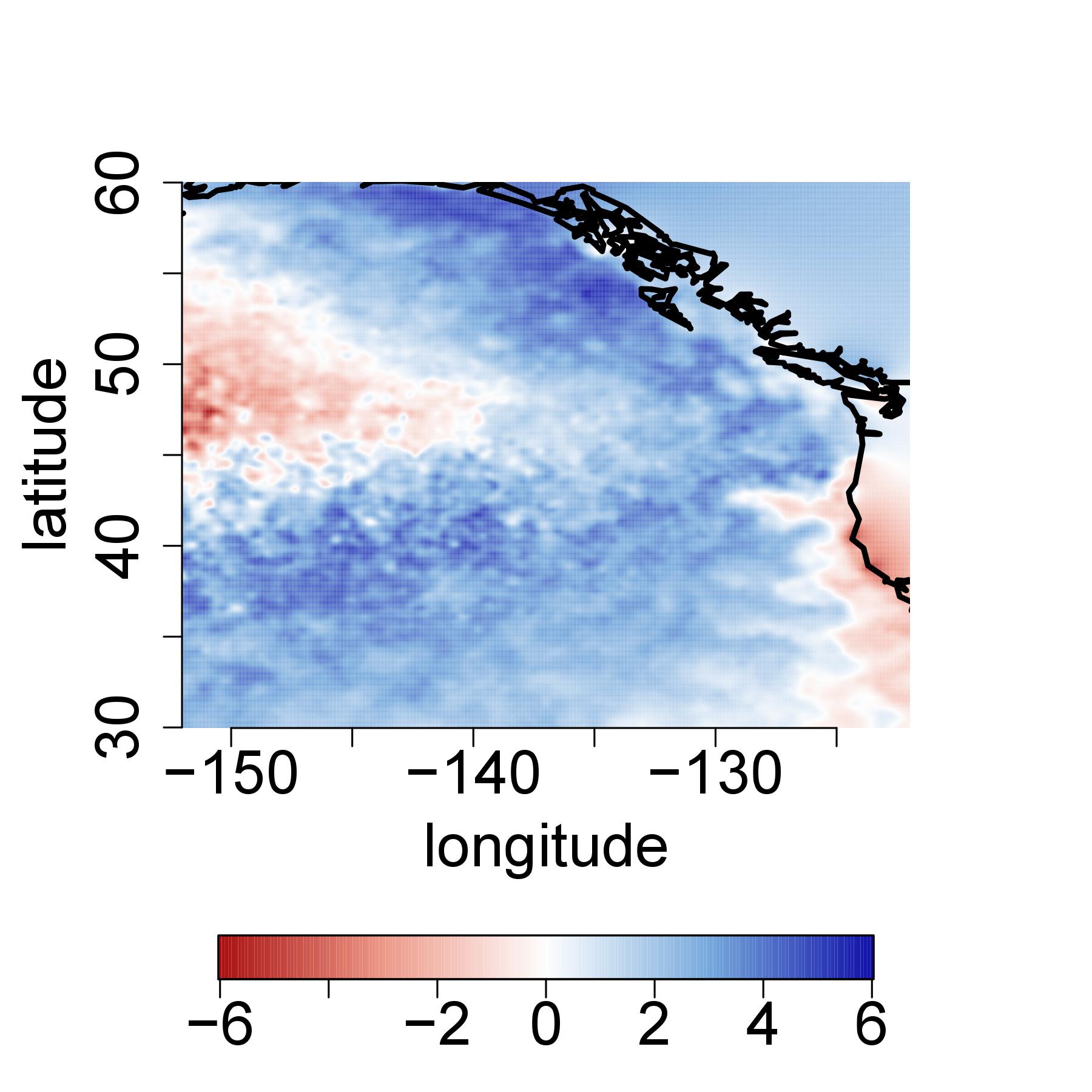}{\label{data-figs:a2}}}
    \subfigure[laGP predicted surface]{\includegraphics[height=5cm,width=5cm]{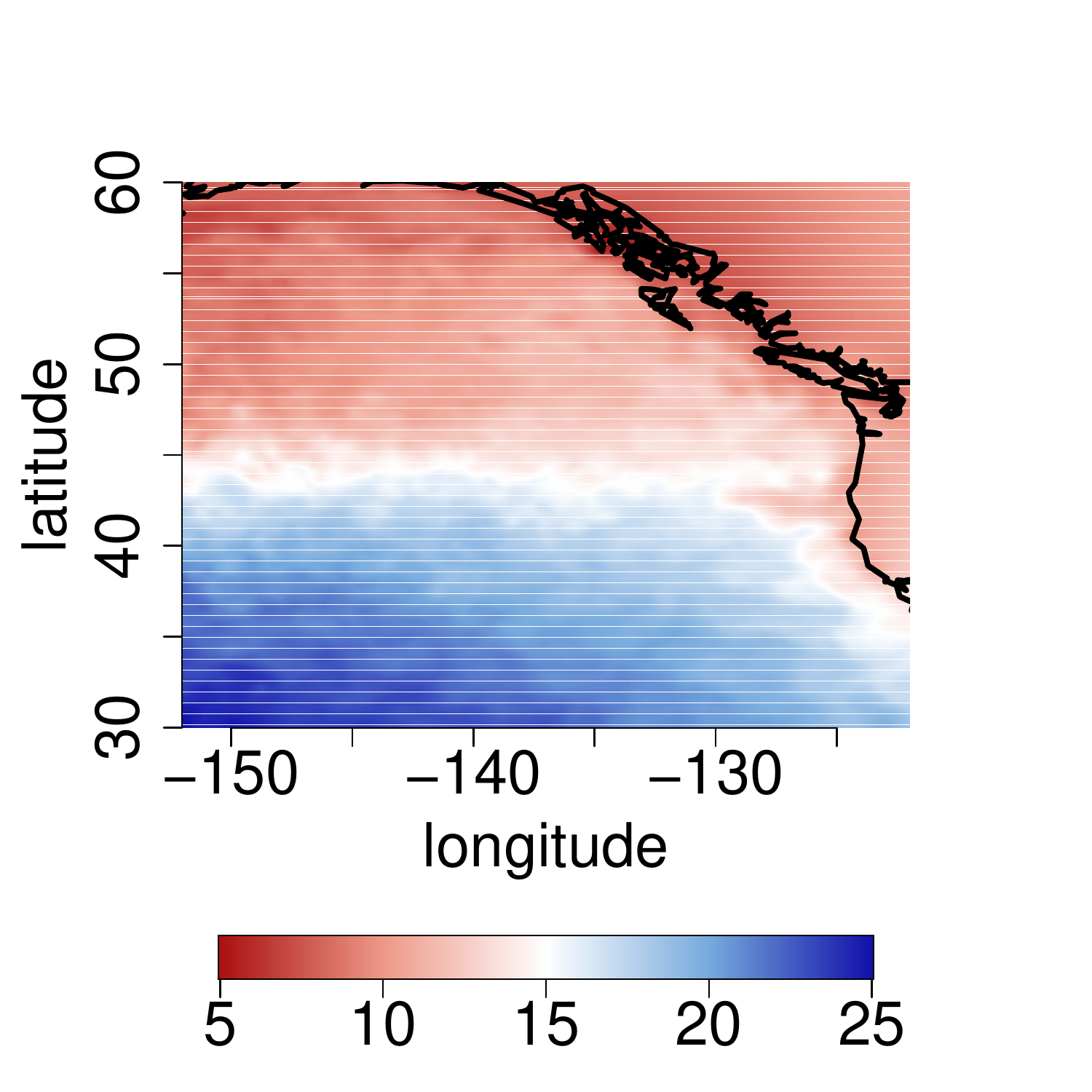}{\label{data-figs:a4}}}
    \subfigure[LatticeKrig predicted surface]{\includegraphics[height=5cm,width=5cm]{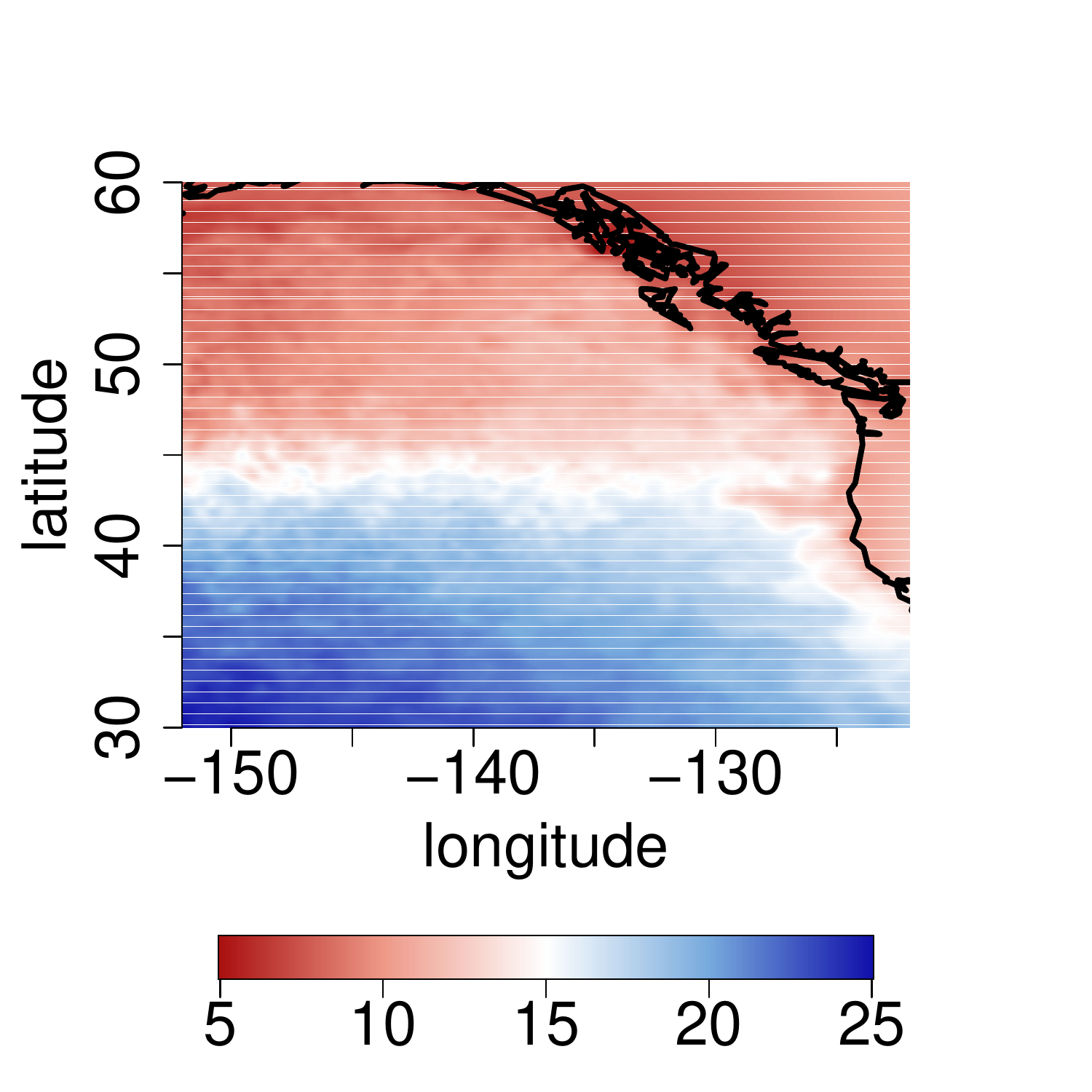}{\label{data-figs:d1}}}\\
    \subfigure[MDCT lower 95\% PI]{\includegraphics[height=5cm,width=5cm]{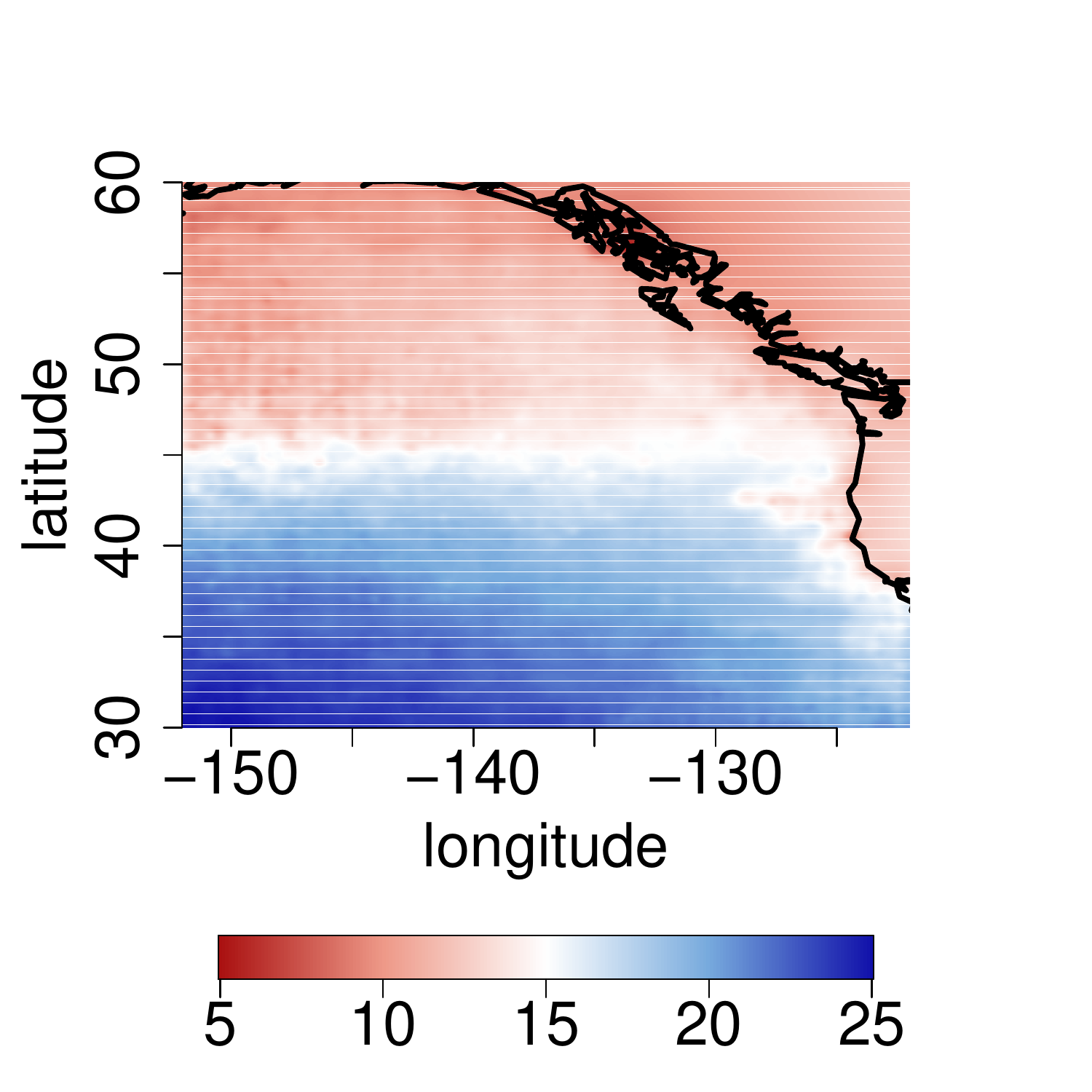}{\label{data-figs:a5}}}
    \subfigure[MDCT upper 95\% PI]{\includegraphics[height=5cm,width=5cm]{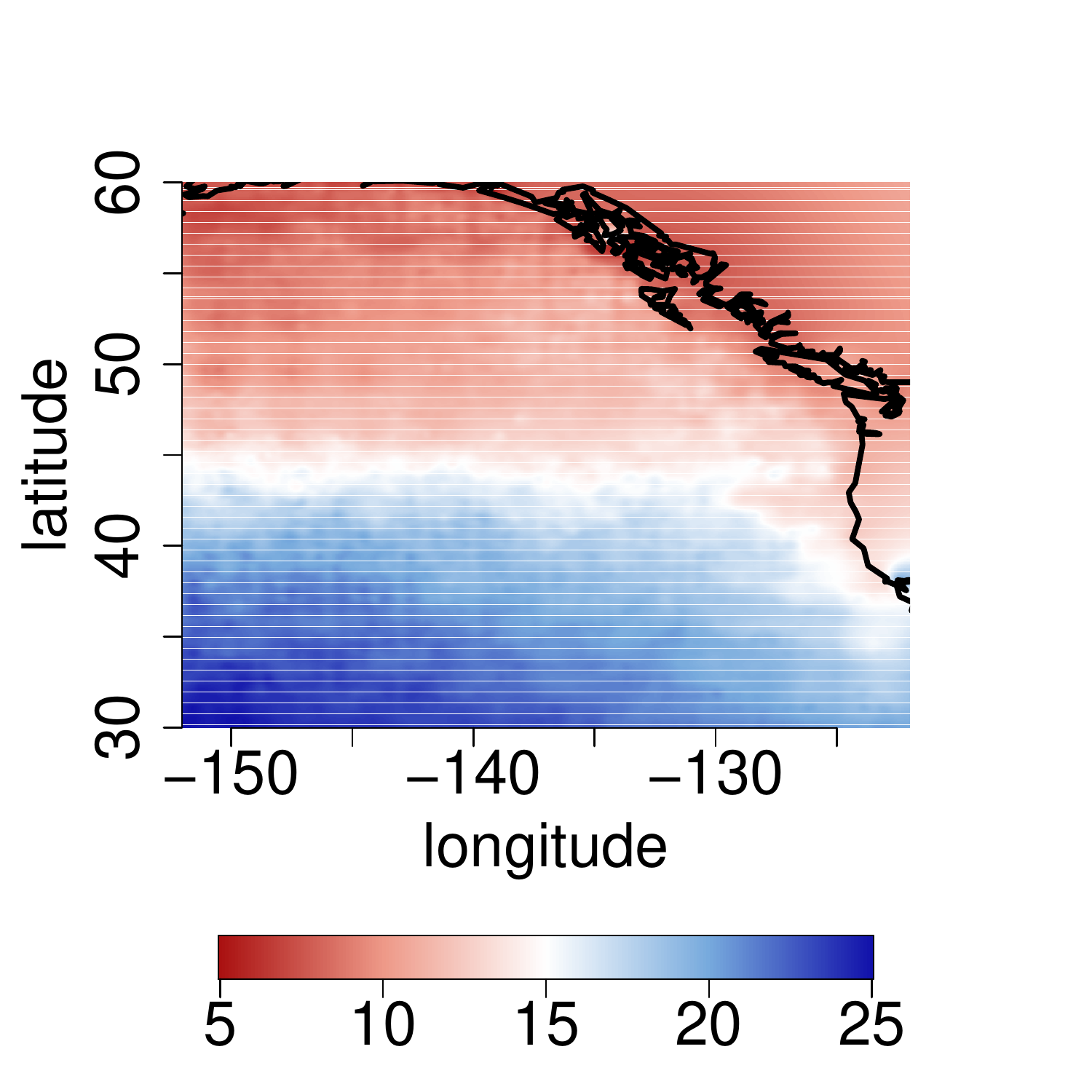}{\label{data-figs:a6}}}
    \subfigure[MDCT predicted surface]{\includegraphics[height=5cm,width=5cm]{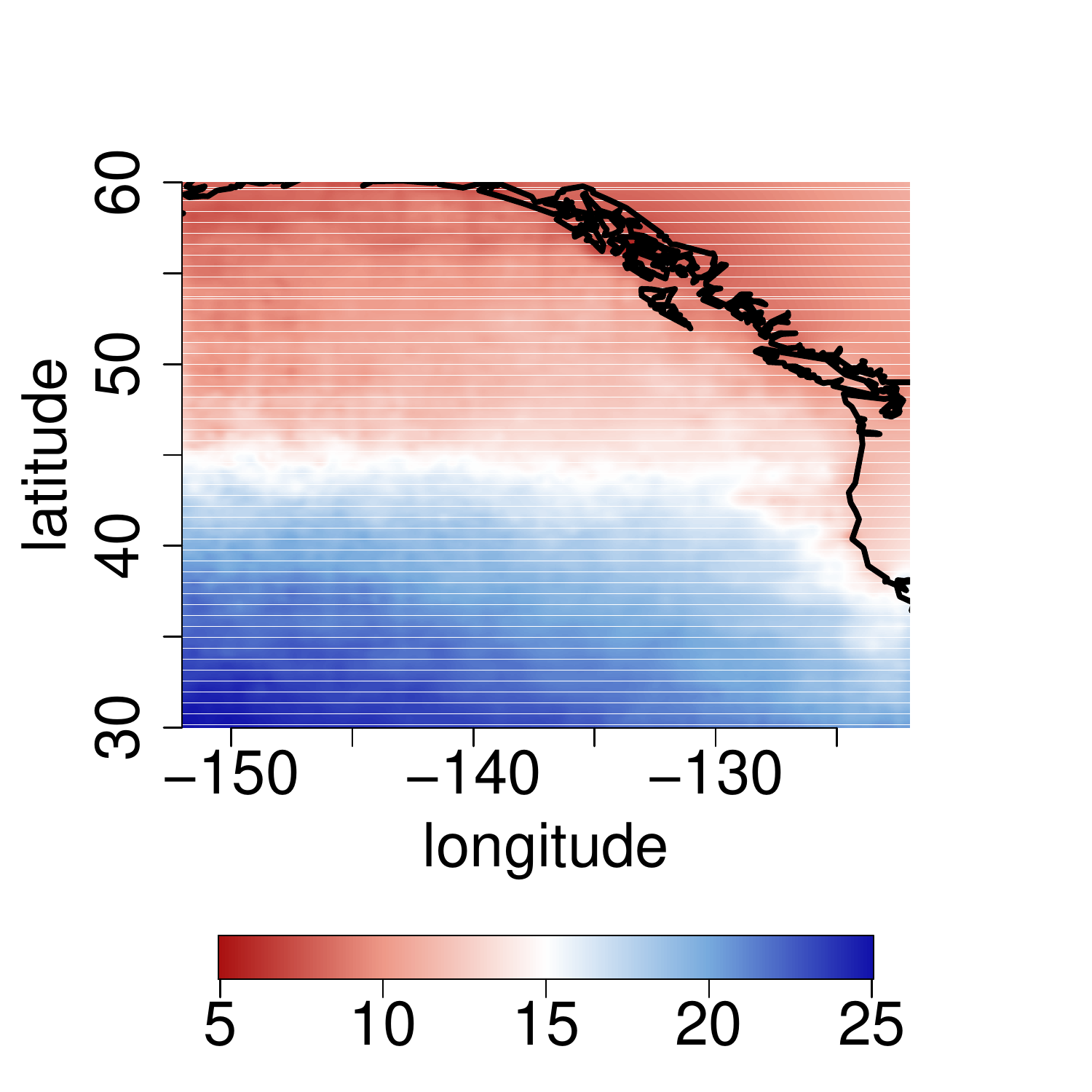}{\label{data-figs:c1}}}
  \end{center}
\caption{(a) Sea surface temperature in October 2016 for a portion of the
North Pacific. (b) Estimated OLS residual from the non-spatial. Panels
(c), (d) and (g) show the estimated mean predictive surfaces for three
competing models. Figures (e) and (f) present the point-wise predictive
bands for the MDCT.}
\label{data}
\end{figure}

The proposed MDCT model for the training data uses $R=3$ resolutions
with the first resolution having $J(1)=100$ knots. To minimize edge
effect, some knots are also kept inside land. Similar to the simulation
studies, we proposed a IG(2,1) prior on $\sigma^2$ and multiscale tree
shrinkage prior on $\beta_j^r$'s with the same hyper-parameters
described in (\ref{multsh}). We implement Algorithm~\ref{post_comp_al}
and run it for $2000$ iterations to discover that $\eta=1$ appearing
overwhelmingly among $2000$ iterations. Thus to reduce unnecessary storage complexity
and computational ease for dataset of this scale, we run the rest of the iterations with $\eta=1$.
The model thereafter is run for $5000$
iterations, convergence diagnostics is performed with the \texttt{coda}
package in R which indicates that $2000$ iterations are sufficient as
burn-in, to achieve practical convergence. As competitors to the MDCT,
we fitted LatticeKrig and LaGP to the data. MPP is computationally
prohibitive for the size of the dataset and is omitted from the
comparison.

\begin{table}[!th]
\begin{center}
\caption{Mean squared prediction error (MSPE), length and coverage of 95\% predictive intervals of MDCT, MDCT(1), MDCT(2), laGP and LatticeKrig. }\label{Tab1:pred_data}
\begin{tabular}
[c]{cccccc}%
\hline
 & MDCT & MDCT(1) & MDCT(2) & laGP & LatticeKrig\\
\hline
MSPE & 0.18 & 0.52 & 0.36 & 0.11 & 0.10\\
Length of 95\% PI & 2.49 & 2.38 & 2.42 & 1.26 & 0.48\\
Coverage of 95\% PI & 0.98 & 0.95 & 0.97 & 0.93 & 0.63\\
\hline
\end{tabular}
\end{center}
\end{table}

  The predictive power of the proposed model, along with that of its
competitors, is assessed based on mean squared prediction error (MSPE),
coverage and length of 95\% predictive intervals. The non-spatial model
and MDCT yield MSPE $1.34$ and $0.18$ respectively. The dramatic improvement in
MSPE due to the inclusion of a spatial structure, that is evident from
Table~\ref{Tab1:pred_data} corroborates the fact that there is a strong
spatial dependence in the field that can not be explained by a linear
effect of longitude and latitude. From the results in
Table~\ref{Tab1:pred_data} we observe that LaGP has a slightly better
performance than MDCT, in terms of MSPE. However, MDCT yields wider
predictive intervals, leading to 98\% coverage, while laGP produces
some under-coverage with a narrower predictive interval. This can be an
indication of overprediction. The smallest MSPE in the table
corresponds to LatticeKrig. Nevertheless, we observe that LatticeKrig
suffers heavily in characterizing predictive uncertainty. Overall, MDCT
turns out to be a competitive performer in predictive inference.
   Predictive surfaces in Figure~\ref{data} further corroborate this
fact. Importantly, even with non-parallel implementation MDCT takes
about $26$ seconds to run one iteration. As shown in Figure~\ref{comp}, the computation
time can be reduced by multiple folds through efficient parallel implementation. On the contrary, even the frequentist implementation
of LatticeKrig takes about 2 hours. Fitting MDCT beyond $R=3$
unnecessarily exacerbates computational burden with minimal improvement
of inferential and predictive performance.

\section{Conclusion}\label{conclusion}
This article proposes a novel multiscale kriging model for spatial datasets.
The model writes the unknown spatial surface as an independent sum of processes
at different scales and is able to approximate a broad class of both 1D and 2D spatial processes
of various degree of smoothness. One key ingredient of our mutiscale model is the kernel
convolution with a compactly supported kernel of minimal degree and knots placed in a regular
grid at every resolution. Theoretically, it allows us to completely characterize the space of functions
generated from the mutiscale spatial model. Another important contribution of the
current article is to propose a new class of \emph{multiscale tree shrinkage prior}
distribution for the basis coefficients. The construction of a tree shrinkage prior
is governed by the consideration that as the model moves to the higher resolutions,
more and more basis coefficients become irrelevant. The idea of multiscale
tree shrinkage prior is novel in statistics and can find applications outside of
spatial analysis. We have been able to show asymptotic convergence properties of the
posterior distribution to rigorously argue consistent surface estimation by the proposed
model.

Besides the important methodological and theoretical contributions that the proposed model
entails, there is an equally important contribution in computational efficiency for massive datasets.
The research on multiscale spatial models is largely motivated by the quest of building
a complex and flexible spatial model that allows accurate spatial inference and prediction
for massive datasets and yet allows rapid Bayesian computation. The compactly supported kernel
together with the multiscale shrinkage priors allow easy MCMC of the model parameters. We develop
a strategy for posterior computation within our modeling paradigm that ensures computation of
the model parameters locally. More specifically, our strategy requires inverting a large number of $[((2d)^R-1)/(2d-1)]\times[((2d)^R-1)/(2d-1)]$
matrices in parallel at every MCMC iteration, leading to unprecedented speed in computation for Bayesian spatial models.

Several interesting new directions open up from this article. First of all, the current framework of mutiscale Bayesian modeling of spatial datasets can readily be extended to spatio-temporal datasets. Secondly, the recent idea of spatial meta kriging (Guhaniyogi and Banerjee, 2016)
allows scalability by fitting a spatial model independently on partitions of a big data followed by combining the inferences. It is established in this article that the proposed multiscale framework can scale up to $\approx$ 1-2 million spatial locations, but may struggle with tens of million. If we have resources to run on $\approx 15$ different subsets, then SMK combined with our approach can yield full Bayesian inference on $\approx$ 30 million locations. Finally, this article proposes one specific rectangular partition of the domain. There is a scope of future research as to how adaptive partitioning of the domain is to be implemented using techniques such as the Voronoi tesselation. Adaptive partitioning with the appropriate placement of knots might significantly reduce the number of knots required to yield acceptably accurate inference. We will explore these approaches in future.

\section*{Appendix}
\textbf{Proof of Theorem~\ref{RKHS}:}

\noindent
Use the fact that $\kappa$ is a compactly supported polynomial of minimal degree for two
dimensions that possesses continuous derivatives upto second order. By Theorem 10.10 and 10.35 in \cite{wendland2004scattered}, we obtain that the Fourier transform of $\kappa$, denoted by $\hat{\kappa}$ satisfies
\begin{align*}
c_1(1+||\omega||_2)^{-d-3}\leq \hat{\kappa}(\omega)\leq c_2(1+||\omega||_2)^{-d-3},
\end{align*}
for some $c_1,c_2>0$. The result now follows using Corollary 10.13 in \cite{wendland2004scattered}.

\noindent\textbf{Proof of Theorem~\ref{consistency}:}

\noindent We begin by stating and proving a lemma that will be useful in the proof of the theorem.
\begin{lemma}\label{prior_positivity}
Consider a ball of radius $\delta$ around $(w_0,\sigma_0^2)$ given by
\begin{align*}
B_{\delta}(w_0,\sigma_0^2)=\left\{(w,\sigma^2): ||w-w_0||_{\infty}<\delta,\left|\frac{\sigma^2}{\sigma_0^2}-1\right|<\delta\right\}.
\end{align*}
Then $\pi(B_{\delta}(w_0,\sigma_0^2))>0$, for all $\delta>0$.
\end{lemma}
\begin{proof}
Since $w_0\in\Theta_c$, $\exists w_{*}(\bs)=\sum_{r=1}^{R^*}\sum_{j=1}^{J(r)}K(\bs,\bs_j^{r*},\phi_r)\beta_{j}^{r*}$, s.t.
$||w_{*}-w_0||_{\infty}<\delta/2$. Note that $K(\cdot,\cdot,\phi_r)$ is a continuous function on a compact set $\mathcal{D}$, implying
$K(\cdot,\cdot,\phi_r)$ to be a uniformly continuous function. Thus, $\exists$ $M$, s.t.
$M=\sup\limits_{\bs\in\mathcal{D}}\max\limits_{r=1,..,R^*;j=1:J(r)}|K(\bs,\bs_j^{r*},\phi_r)|$. Assume further that $\eta=\sum_{r=1}^{R^*}\sum_{j=1}^{J(r)}|\beta_{j}^{r*}|$. Since $K$ is uniformly continuous, one can choose $\bs_j^{r}$'s such that
$\sup\limits_{\bs\in\mathcal{D}}|K(\bs,\bs_j^r,\phi_r)-K(\bs,\bs_j^{r*},\phi_r)|<\frac{\delta}{4\eta\sum_{r=1}^{R^*}J(r)}$. Define the set
\begin{align*}
\mathcal{I}=\left\{\{\beta_j^r\}: |\beta_j^r-\beta_j^{r*}|<\frac{\delta}{4M\sum_{r=1}^{R^*}J(r)}\right\}.
\end{align*}
Clearly, for the set of all $w(\bs)=\sum_{r=1}^{R^*}\sum_{j=1}^{J(r)}K(\bs,\bs_j^{r},\phi_r)\beta_{j}^{r}$, with $\bs_j^{r}$ is chosen as above and $\beta_j^r$ chosen from $\mathcal{I}$, we have
\begin{align*}
|w_0(\bs)-w(\bs)| &\leq |w_0(\bs)-w_{*}(\bs)|+|w_{*}(\bs)-w(\bs)|\\
&\leq\frac{\delta}{2}+\sum_{r=1}^{R^*}\sum_{j=1}^{J(r)}|K(\bs,\bs_j^r,\phi_r)||\beta_j^r-\beta_j^{r*}|+\sum_{r=1}^{R^*}\sum_{j=1}^{J(r)}|\beta_j^{r*}|
|K(\bs,\bs_j^r,\phi_r)-K(\bs,\bs_j^{r*},\phi_r)|\\
&\leq \frac{\delta}{2}+\frac{\delta M\sum_{r=1}^{R^*}J(r)}{4M\sum_{r=1}^{R^*}J(r)}+\frac{\delta \eta\sum_{r=1}^{R^*}J(r)}{4\eta\sum_{r=1}^{R^*}J(r)}=\delta.
\end{align*}
Thus $\mathcal{I}\times\left\{\sigma^2:\left|\frac{\sigma^2}{\sigma_0^2}-1\right|<\delta\right\}\subseteq B_{\delta}(w_0,\sigma_0^2)$. Since, the prior on all $\beta_j^r$ are continuous on the entire real line and the prior on $\sigma^2$ is also continuous on $\mathcal{R}^{+}$, it trivially holds that $\pi(B_{\delta}(w_0,\sigma_0^2))\geq \pi\left(\mathcal{I}\times\left\{\sigma^2:\left|\frac{\sigma^2}{\sigma_0^2}-1\right|<\delta\right\}\right)>0$. This concludes the proof of the lemma.
\end{proof}
We will now proceed with the proof of Theorem~\ref{consistency}. Our aim is to check that all conditions of Theorem in \cite{choi2007posterior} are satisfied. Let $H_i=\frac{N\left(y_i|w_0(\bs_i),\sigma_0^2\right)}{N\left(y_i|w(\bs_i),\sigma^2\right)}$, and $K_i(w,w_0)=\mbox{E}(H_i)$ and $V_i(w,w_0)=\mbox{Var}(H_i)$. It is easy to check that (\cite{choi2007posterior})
\begin{align*}
K_i(w,w_0)&=\frac{1}{2}\log\frac{\sigma^2}{\sigma_0^2}-\frac{1}{2}\left(1-\frac{\sigma_0^2}{\sigma^2}\right)+\frac{1}{2}\frac{(w(\bs)-w_0(\bs))^2}{\sigma^2}\\
V_i(w,w_0) &=\frac{1}{2}\left(\frac{\sigma_0^2}{\sigma^2}-1\right)^2+\frac{\sigma_0^4}{\sigma^4}(w(\bs)-w_0(\bs))^2.
\end{align*}
Thus for every $\epsilon>0$, there exists a $\delta>0$ such that $(w(\cdot),\sigma^2)\in B_{\delta}(w_0,\sigma_0^2)$ implies
$K_i(w,w_0)<\epsilon,\forall\:i$ and $\sum_{i=1}^{\infty}\frac{V_i(w,w_0)}{i^2}<\infty$. Thus condition (i) is satisfied.
Condition (ii), i.e. the prior positivity has already been proved to be satisfied by Lemma~\ref{prior_positivity}.

Finally, the condition of having an exponentially consistent sequence of tests follows along the same line as the proof of Theorem 2 in \cite{choi2007posterior}. This concludes the theorem.
\newpage
\bibliographystyle{chicago}
\bibliography{bibliography_rev}

\end{document}